\documentclass[12pt]{article}
\usepackage{amsfonts,amssymb,amsbsy,amsmath,amsthm,enumerate}

\usepackage{color}

\newtheorem{theorem}{Theorem}[section]

\newtheorem{follow}[theorem]{Corollary}

\newtheorem{pr}[theorem]{Proposition}
\theoremstyle{definition}

\newcommand{\bel}{\begin{equation} \label}
\newcommand{\ee}{\end{equation}}

\newcommand{\rd}{{\mathbb R}^{2}}
\newcommand{\re}{{\mathbb R}}

\newcommand{\N}{{\mathbb N}}
\newcommand{\ce}{{\mathbb C}}

\def\beq{\begin{equation}}
\def\eeq{\end{equation}}
\newcommand{\bea}{\begin{eqnarray}}
\newcommand{\eea}{\end{eqnarray}}
\newcommand{\beas}{\begin{eqnarray*}}
\newcommand{\eeas}{\end{eqnarray*}}

{

\begin{document}
\begin{center}
{\Large \bf  Discrete Spectrum of Quantum Hall Effect Hamiltonians
I. Monotone Edge Potential}

\medskip

{\sc Vincent Bruneau, Pablo Miranda, Georgi Raikov}

\medskip

\today
\end{center}

\bigskip

{\bf Abstract.} {\small We consider the unperturbed operator $H_0
: = (-i\nabla - {\bf A})^2 + W$, self-adjoint in $L^2(\rd)$. Here
$A$ is a magnetic potential which generates a constant magnetic
field $b>0$, and the edge potential $W$ is a non-decreasing non
constant bounded function depending only on the first coordinate
$x \in \re$ of $(x,y) \in \rd$.  Then the  spectrum of $H_0$ has a
band structure and is absolutely continuous; moreover, the
assumption $\lim_{x \to \infty}(W(x) - W(-x)) < 2b$ implies the
existence of infinitely many spectral gaps for $H_0$. We consider
the perturbed operators $H_{\pm} = H_0 \pm V$ where the electric
potential $V \in L^{\infty}(\rd)$ is non-negative and decays at
infinity. We investigate the asymptotic distribution of the
discrete spectrum of $H_\pm$ in the spectral gaps of $H_0$. We
introduce an effective Hamiltonian which governs the main
asymptotic term; this Hamiltonian involves a pseudo-differential
operator with generalized anti-Wick symbol equal to $V$. Further,
we restrict our attention on perturbations $V$ of compact support
and constant sign. We establish a geometric condition on the
support of $V$ which guarantees the finiteness of the eigenvalues
of $H_{\pm}$ in any spectral gap of $H_0$. In the case where this
condition is violated, we show that, generically, the convergence
of the infinite series of eigenvalues of $H_+$ (resp. $H_-$) to
the left (resp. right) edge of a given spectral gap, is
Gaussian.}\\



\section{Introduction}
\label{section1} \setcounter{equation}{0}

    The general form of the unperturbed operators we are going to
    consider in the present article and its eventual second part, is
    $$
    H_0 = H_0(b,W) : = -\frac{\partial^2}{\partial x^2} +
    \left(-i\frac{\partial}{\partial y} - bx\right)^2 + W(x).
    $$
Here $b>0$ is the constant magnetic field, and the edge potential
$W \in L^{\infty}(\re; \re)$ is independent of $y$. The
self-adjoint operator $H_0$ is defined initially on
$C_0^{\infty}(\rd)$ and then is closed in $L^2(\rd)$.  Let
${\mathcal F}$ be the partial Fourier transform with respect to
$y$, i.e.
$$
({\mathcal F}u)(x,k) = (2\pi)^{-1/2} \int_\re e^{-iyk} u(x,y) dy,
\quad u \in L^2(\rd).
$$
Then we have
$$
{\mathcal F} H_0 {\mathcal F}^* = \int_\re^{\oplus} h(k) dk
$$
where the operator
$$
h(k) : = - \frac{d^2}{dx^2} + (bx-k)^2 + W(x), \quad k \in \re,
$$
is self-adjoint in $L^2(\re)$. For $w \in L^2(\re)$ and $k \in
\re$ set $(\tau_k w)(x) : = w(x-k/b)$. Evidently $\tau_k$ is a
unitary operator in $L^2(\re)$, and we have $\tau_k^* h(k) \tau_k
= \tilde{h}(k)$ where
$$
\tilde{h}(k) : = - \frac{d^2}{dx^2} + b^2 x^2  +  W(x + k/b),
\quad k \in \re.
$$
Evidently, for each $k \in \re$ the operator $h(k)$ (and, hence,
$\tilde{h}(k)$) has a discrete and simple spectrum. Let
$\left\{E_j(k)\right\}_{j=1}^{\infty}$ be the increasing sequence
of the eigenvalues of $h(k)$ (and, hence, of $\tilde{h}(k)$). The
Kato analytic perturbation theory \cite{K} implies that $E_j(k)$,
$j \in {\mathbb N}$, are real analytic functions of $k \in \re$.
When we need to indicate the dependence of $E_j(k)$  on $b$ and/or
$W$, we will  write $E_j(k;b,W)$ or $E_j(k;W)$ instead of
$E_j(k)$. Note that if $W=0$, then the eigenvalues are independent
of $k$, and their explicit form is well-known:
$$
E_j(k;b,0) = E_j(b,0) = b (2j-1), \quad k \in \re, \quad j \in {\mathbb N}.
$$
Further,
    \bel{1}
    \sigma(H_0) = \bigcup_{j=1}^{\infty} \overline{E_j(\re)}.
    \ee
In the present article we will consider monotone $W$. For
definiteness we assume that $W$ is non-decreasing.
    Then the band functions $E_j$, $j \in {\mathbb N}$, are also non-decreasing, and {$\sigma(H_0) = \bigcup_{j=1}^{\infty} [{\mathcal E}_j^-, {\mathcal E}_j^+]$ with}
 \bel{6}
   {\mathcal E}_j^- =  \lim_{k \to -\infty} E_j(k) = b(2j-1) + W_-, \quad
     {\mathcal E}_j^+ = \lim_{k \to \infty} E_j(k) = b(2j-1) + W_+,
    \ee
    $$
 {   W_- := \lim_{x \to -\infty} W(x), \quad W_+ := \lim_{x \to \infty}  W(x).}
    $$
    (see Proposition \ref{p21} below).
    Throughout the article we assume that $W_- < W_+$, i.e.
$W$ is not identically constant. Hence, ${\mathcal E}_j^- <
{\mathcal E}_j^+$ for each $j \in {\mathbb N}$, which implies that
the spectrum of $H_0$ is absolutely continuous. Moreover, we will assume
that
    \bel{5}
    W_+ - W_- < 2b.
    \ee
    Then we have
    \bel{4}
{\mathcal E}_j^+ < {\mathcal E}_{j+1}^-, \quad j \in {\mathbb N},
    \ee
and the intervals $({\mathcal E}_j^+, {\mathcal E}_{j+1}^-)$, $j
\in {\mathbb N}$, are open gaps in the spectrum of $H_0$. Thus all
the  bands in the spectrum of $H_0$ are separated by gaps where discrete
spectrum may appear under appropriate perturbations. \\
The perturbations under consideration will be electric potentials
$V : \rd \to \re $ which are $\Delta$-compact. A simple sufficient
condition which guarantees the compactness of the operator $V
(-\Delta - i)^{-1}$, is
$$
V \in L^{\infty}_0(\rd) : = \{u \in L^{\infty}(\rd) \, | \, u(x,y)
\to 0 \; {\rm as} \; x^2 + y^2 \to \infty\}.
$$
By the diamagnetic inequality, the operator $V (H_0 - i)^{-1}$ is
also compact, and hence
$$
\sigma_{\rm ess}(H_0 + V) = \sigma_{\rm ess}(H_0) =
\bigcup_{j=1}^{\infty} [{\mathcal E}_j^-, {\mathcal E}_j^+].
$$
 For simplicity, we will consider perturbations of definite
sign. More precisely we will suppose that $V \geq 0$, and will
consider the operators $H_{\pm} : = H_0 \pm V$. Note that in the
case of positive (resp. negative) perturbations, the discrete
eigenvalues of the perturbed operator which may appear in a given
open gap of the spectrum of the unperturbed operator, may
accumulate only to the left (resp. right) end of the gap. \\
In order to give a more explicit formulation of the problem, we
need the following notations.  Let $T$ be a self-adjoint linear
operator in a Hilbert space. Denote by ${\mathbb P}_{\mathcal
O}(T)$ the spectral projection of $T$ corresponding to the Borel
set ${\mathcal O} \subseteq \re$. For $\lambda > 0$ set
$$
{\mathcal N}_0^-(\lambda) : = {\rm rank}\,{\mathbb P}_{(-\infty,
{\mathcal E}^-_1-\lambda)}(H_-).
$$
Next, fix $j \in {\mathbb N}$ and assume that \eqref{5} holds.
Pick $\lambda \in (0, {\mathcal E}^-_{j+1}-{\mathcal E}^+_j)$, and
set
$$
{\mathcal N}_j^-(\lambda) : = {\rm rank}\,{\mathbb P}_{({\mathcal
E}^+_j, {\mathcal E}^-_{j+1}-\lambda)}(H_-), \quad
{\mathcal N}_j^+(\lambda) : = {\rm rank}\,{\mathbb P}_{({\mathcal
E}^+_j+\lambda, {\mathcal E}^-_{j+1})}(H_+).
$$
We reduce the investigation of the accumulation of the discrete
eigenvalues to the edges of the gap $({\mathcal E}^+_j, {\mathcal
E}^-_{j+1})$ of its essential spectrum,  to the study of the
asymptotic
behavior as $\lambda \downarrow 0$ of the counting functions ${\mathcal N}_j^\pm(\lambda)$.\\
The investigation of the asymptotic behavior of the discrete
spectrum of perturbed analytically fibered quantum Hamiltonians,
lying in the gaps of the essential one has a long history.
Probably, the first results of this type were obtained for the
Schr\"odinger operator with periodic potential perturbed by a
decaying one (see e. g. \cite{z, khr, r2, sch}). Recently, similar
problems have been considered for perturbed 2D magnetic
Hamiltonians \cite{brs}, and for Dirichlet Laplacians in twisted
wave-guides \cite{bkrs}. The common feature of the above cited
articles is that the edges of the gaps in the spectrum of the
unperturbed operator coincide with  the extremal values of the
band functions taken at local non degenerate extrema; in this case
the arising effective Hamiltonian is a differential Schr\"odinger
type operator. In the present article the edges of the gaps are
the limiting values of the band functions $E_j(k)$ as $k \to \pm
\infty$. {The effective Hamiltonian} which arises in this case
involves a ``kinetic" part equal in the momentum representation to the multiplier by $E_j$, and a
``potential" part which is a pseudodifferential operator
($\Psi$DO) with contravariant (generalized anti-Wick) symbol equal
to $V$.  These $\Psi$DOs are unitarily equivalent to the
Berezin-Toeplitz operators which appear as effective Hamiltonians
in the study of compact perturbations of the Landau operator (see
e. g. \cite{r1, rw}). Note however that in the case of the Landau
operator (which is equal to $H_0$ with $W=0$) the effective
Hamiltonian has no kinetic part.

 The article is organized as follows. In Section
\ref{section2} we describe the basic spectral properties of the
unperturbed operators which we need in the sequel. In Section
\ref{section3} we introduce the effective Hamiltonian appropriate
for the asymptotic analysis as $\lambda \downarrow 0$ of the
function ${\mathcal N}_j^{\pm}(\lambda)$ with fixed $j \in
{\mathbb N}$. Our effective Hamiltonian approach allows us to
consider various types of $W$ and $V$ which satisfy the
assumptions stated above. Nonetheless, the rest of the article is
dedicated to the case where $V \in L^{\infty}_0(\rd;\re)$ has a
compact support. This choice is motivated by the possible
applications in the theory of the quantum Hall effect (see e. g.
\cite{chs, hs, cg}{)}, and, on the other hand, by the spectacular
progress in the investigation of the discrete spectrum for
localized perturbations of the Landau Hamiltonian $H_0(b,0)$ (see
e. g. \cite{rw, mroz, proz, fp, rozt, pe, proz2}). For
definiteness, we suppose that $V \geq 0$ and discuss only the
behavior of the counting functions ${\mathcal N}_j^{+}(\lambda)$,
$j \in {\mathbb N}$, near the lower edges of the spectral gaps; in
the case $V \leq 0$ the behavior of ${\mathcal N}_j^{-}(\lambda)$
near the upper edges is analogous. In Section \ref{section4} we
establish a sufficient condition of geometric nature which
guarantees that all the functions ${\mathcal N}_j^{+}(\lambda)$,
$j \in {\mathbb N}$, remain bounded as  $\lambda \downarrow 0$,
i.e. that there is a finite number of eigenvalues of $H_+$ in any
gap of its essential spectrum. When this sufficient condition is
violated, we show that for any $j \in {\mathbb N}$ the functions
${\mathcal N}_j^{+}(\lambda)$ generically blow up as $\lambda
\downarrow 0$. More precisely, in Section \ref{section5}
we reduce the analysis of ${\mathcal N}_j^{+}(\lambda)$ to counting
functions for operators in holomorphic spaces. These operators are
studied in Section \ref{section6} in order to
establish a  lower asymptotic estimate
    \bel{lau2}
    {\mathcal C}_-  |\ln{\lambda}|^{1/2} (1 + o(1)) \leq  {\mathcal
N}_j^{+}(\lambda), \quad \lambda \downarrow 0, \quad j \in
{\mathbb N},
    \ee
    with ${\mathcal C}_- > 0$ which holds when the sufficient condition of Section \ref{section4} is not fulfilled,
    and  an upper
asymptotic estimate
    \bel{lau3}
 {\mathcal N}_j^{+}(\lambda) \leq {\mathcal C}_+  |\ln{\lambda}|^{1/2} (1 + o(1)), \quad \lambda \downarrow 0,
 \quad j \in {\mathbb N},
    \ee
with ${\mathcal C}_+ > {\mathcal C}_-$. Note that the constants
${\mathcal C}_{\pm}$ in \eqref{lau2} and \eqref{lau3} admit a
clear geometric interpretation and are independent of $j \in
{\mathbb N}$. Thus, in the case of infinitely many eigenvalues in
any given gap, the main asymptotic term of ${\mathcal
N}_j^{+}(\lambda)$ is expected to be of order
$|\ln{\lambda}|^{1/2}$ which, loosely speaking, corresponds to a
Gaussian convergence of the discrete eigenvalues to the edges of
the gaps of the essential spectrum. This behavior is different
from the case of compactly supported perturbations of the Landau
Hamiltonian where typically we have
    \bel{lau4}
    {\mathcal N}_j^{+}(\lambda) \sim
\frac{|\ln{\lambda}|}{\ln{|\ln{\lambda}|}}, \quad \lambda
\downarrow 0, \quad j \in {\mathbb N},
    \ee
(see e.g. \cite{rw}). Hopefully, in a future work we will attack
the problem of finding the main asymptotic term as $\lambda
\downarrow 0$ of ${\mathcal N}_j^{\pm}(\lambda)$, $j \in {\mathbb
N}$.

\section{Basic spectral properties of $H_0$}
\label{section2} \setcounter{equation}{0}
In the following proposition we consider the general properties of the band functions $E_j$, $j \in {\mathbb N}$.
By analogy with the operator $\tilde{h}(k)$, introduce the shifted harmonic oscillator
$$
\tilde{h}_{\infty} : = -\frac{d^2}{dx^2} + b^2 x^2 +W_+,
$$
 which is
self-adjoint in $L^2(\re)$, and essentially self-adjoint on
$C_0^{\infty}(\re)$.
    \begin{pr} \label{p21}
    Assume that $W$ is  non-decreasing and bounded. Then for
    each $j \in {\mathbb N}$ the eigenvalue $E_j(k)$ is a
    non-decreasing function of $k \in \re$, and  \eqref{6} holds true.
\end{pr}
    \begin{proof}
    The fact that $E_j$ are non-decreasing bounded functions of $k$
    follows directly from the mini-max principle. Let us prove (\ref{6}).
Pick $E > - b - W_-$. Then for each $k \in \re$ we have
$
-E < b + W_- \leq \inf \sigma(\tilde{h}(k))$.
Moreover,
$-E < b + W_+ = \inf \sigma(\tilde{h}_{\infty})$.
 Then,
    $$
    |(E_j(k) + E)^{-1} - (b(2j-1) + W_+ + E)^{-1}| \leq
    $$
    $$
    \|(\tilde{h}(k)+E)^{-1}(W_+ - W(\cdot +
    k/b))(\tilde{h}_{\infty}+E)^{-1}\| \leq
    $$
    \bel{8}
\|(\tilde{h}(k)+E)^{-1}\| \|(W_+ - W(\cdot +
    k/b))(\tilde{h}_{\infty}+E)^{-1}\|.
    \ee
    Moreover,
    \bel{9}
\|(\tilde{h}(k)+E)^{-1}\| \leq (E+b+W_-)^{-1},
    \ee
    and the r.h.s. is $k$-independent. Further, the multiplier by
$(W_+ - W(\cdot + k/b))$, $x \in\re$, tends strongly to zero as $k
\to \infty$, while the operator $(\tilde{h}_{\infty}+E)^{-1}$ is
compact and $k$-independent. Hence, the operator $(W_+ - W(\cdot +
    k/b))(\tilde{h}_{\infty}+E)^{-1}$ tends uniformly to zero as $k \to
    \infty$. Now, \eqref{8} -- \eqref{9} imply
    $$
    \lim_{k \to \infty}(E_j(k) + E)^{-1}= (b(2j-1) + W_+ + E)^{-1},
    \quad j \in {\mathbb N},
    $$
    which  yields the second limit in \eqref{6}. The first one is proved in  the same manner.
        \end{proof}
Our next theorem will play a crucial role in the construction of
the effective Hamiltonian introduced in the next section. For its
formulation we need the following notations.  Fix $k \in \re$ and
$j \in {\mathbb N}$ denote by $\re \ni x \mapsto \psi_j(x;k) \in
\re$ the eigenfunction of the operator $h(k)$ which satisfies
    \bel{defpsij}
h(k) \psi_j(x;k) = E_j(k) \psi_j(x;k), \quad
\|\psi_j(\cdot;k)\|_{L^2(\re)} = 1. \ee Set \bel{lau5} \pi_j(k) :
= \langle \cdot, \psi_j(\cdot;k)\rangle \psi_j(\cdot;k),
    \ee
where $\langle\cdot,\cdot\rangle$ denotes the scalar product in
$L^2(\re)$. Evidently, $\pi_j(k)$ is a rank-one orthogonal
projection acting in $L^2(\re)$.

Fix $j \in {\mathbb N}$ and $k \in \re$, and denote by $\re \ni x
\mapsto \psi_{j,\infty}(x;k) \in \re$ the eigenfunction which
satisfies
    \bel{lau40}
-\psi_{j,\infty}(x;k)^{''} + (bx-k)^2 \psi_{j,\infty}(x;k) =
b(2j-1) \psi_{j,\infty}(x;k), \quad
\|\psi_{j,\infty}(\cdot;k)\|_{L^2(\re)} = 1.
    \ee
The functions $\psi_{j,\infty}$, $j \in {\mathbb N}$, admit a
simple explicit description. Let $\varphi_j$ be the real-valued
eigenfunction satisfying
$$
-\varphi_j''(x) + x^2 \varphi(x) = (2j-1) \varphi_j(x), \quad
\|\varphi_j\|_{L^2(\re)} = 1.
$$
We have
    \bel{61}
\varphi_j(x) =  \frac{H_{j-1}(x) e^{-x^2/2}}{(\sqrt{\pi}2^{j-1}
(j-1)!)^{1/2}}, \quad x \in \re, \quad j \in {\mathbb N},
    \ee
where
$$
{\rm H}_q(x): = (-1)^q e^{x^2} \frac{d^q}{dx^q} e^{-x^2}, \quad x
\in \re, \quad q \in {\mathbb Z}_+,
$$
are the Hermite polynomials. Then
    \bel{sof9}
\psi_{j,\infty}(x;k) = b^{1/4}\varphi_j(b^{1/2}x - b^{-1/2}k), \quad j
\in {\mathbb N}, \quad x \in \re, \quad k \in \re.
    \ee
    Put
    \bel{sof5}
    p_j = p_j(b) : = \frac{ b^{-j+3/2}}{\sqrt{\pi} (j-1)!2^{j-1}}, \quad j \in {\mathbb N}.
    \ee
    Note that we have
    \bel{sof6}
    \psi_{j,\infty}(x;k) = p_j^{1/2} (-k)^{j-1} e^{-(b^{-1/2}k - b^{1/2}x)^2/2}(1+o(1))
     \ee
     as $k \to \infty$, uniformly with respect to $x$ belonging to compact subset of $\re$.
Set
$$
\pi_{j,\infty}(k) : = \langle \cdot,
\psi_{j,\infty}(\cdot;k)\rangle \psi_{j,\infty}(\cdot;k).
    $$

\begin{theorem} \label{teth1}
Fix $j \in {\mathbb N}$. Then we have
    \bel{lau6}
    \lim_{k \to \infty} \left({\mathcal E}_j^+ - E_j(k)\right)^{-1/2} \|\pi_{j,\infty} -
    \pi_j(k)\|_1 = 0
    \ee
    where $\|T\|_1$ denotes the trace-class norm of the operator $T$.
    \end{theorem}

    We will divide the proof of the theorem into several lemmas and propositions.\\
Set
$$
\tilde{\psi}_j(x;k) : = \psi_j(x+b^{-1}k;k), \quad x \in \re,
\quad k \in \re, \quad j \in {\mathbb N},
$$
(see \eqref{defpsij} for the definition of the function $\psi_j$).
Evidently,
$$
\tilde{h}(k) \tilde{\psi}_j(x;k) = E_j(k) \tilde{\psi}_j(x;k), \quad
\|\tilde{\psi}_j(\cdot;k)\|_{L^2(\re)} = 1.
$$
By analogy with \eqref{lau5} put
$$
\tilde{\pi}_j(k) : = \langle \cdot, \tilde{\psi}_j(\cdot;k)\rangle
\tilde{\psi}_j(\cdot;k),  \quad k \in \re, \quad j \in {\mathbb N}.
$$
Similarly, set
$$
\tilde{\psi}_{j,\infty}(x) : = b^{1/4}\varphi_j(b^{1/2}x), \quad x \in \re,  \quad j \in {\mathbb N},
$$
(see \eqref{61} for the definition of the function $\varphi_j$).
Then
$$
\tilde{h}_{\infty} \tilde{\psi}_{j,\infty}(x) = {\mathcal E}^+_j
\tilde{\psi}_{j,\infty}(x), \quad
\|\tilde{\psi}_{j,\infty}\|_{L^2(\re)} = 1.
$$
Put
$$
\tilde{\pi}_{j,\infty} : = \langle \cdot,
\tilde{\psi}_{j,\infty}\rangle \tilde{\psi}_{j,\infty},  \quad j
\in {\mathbb N}.
$$

Since we have
$$
\tau_k \tilde{\pi}_j(k) \tau_k^* = \pi_j(k), \quad \tau_k \tilde{\pi}_{j,\infty} \tau_k^* = \pi_{j,\infty}(k),
\quad k \in \re, \quad j \in {\mathbb N},
$$
relation \eqref{lau6} is equivalent to
 \bel{te1}
    \lim_{k \to \infty} \left({\mathcal E}_j^+ - E_j(k)\right)^{-1/2} \|\tilde{\pi}_{j,\infty} -
    \tilde{\pi}_j(k)\|_1 = 0.
    \ee
For $z \in {\mathbb
C}\setminus(\sigma(\tilde{h}_{\infty})\setminus\{{\mathcal
E}_j^+\})$ set
$$
R_{0,j}^{\perp}(z) : = (\tilde{h}_{\infty} - z)^{-1}
(I-\tilde{\pi}_{j,\infty}).$$ Similarly, for $z \in {\mathbb
C}\setminus(\sigma(\tilde{h}(k))\setminus\{E_{j}(k)\})$ put
$$
R_j^{\perp}(z) : = (\tilde{h}(k) - z)^{-1} (I-\tilde{\pi}_{j}(k)).
$$
Set
$$
U_k(x) : = W_+ - W(x + k/b)= \tilde{h}_{\infty} - \tilde{h}(k),
\quad x \in \re, \quad k \in \re.
$$
\begin{pr} \label{tep1}
We have
    \bel{te4_0}
    \tilde{\pi}_{j,\infty}= \tilde{\pi}_{j,\infty} \tilde{\pi}_j(k) -  \tilde{\pi}_{j,\infty} U_k R_j^{\perp}({\mathcal E}_j^+)
=  \tilde{\pi}_j(k)\tilde{\pi}_{j,\infty} - R_j^{\perp}({\mathcal
E}_j^+) U_k \tilde{\pi}_{j,\infty},
    \ee
    \bel{te4_k}
 \tilde{\pi}_j(k)= \tilde{\pi}_j(k)\tilde{\pi}_{j,\infty}  + \tilde{\pi}_j(k) U_k R_{0,j}^{\perp}(E_{j})
= \tilde{\pi}_{j,\infty} \tilde{\pi}_j(k) + R_{0,j}^{\perp}(E_{j})
U_k \tilde{\pi}_j(k).
    \ee
    \end{pr}
    \begin{proof}
We have
$$
 \tilde{\pi}_{j,\infty} = \tilde{\pi}_{j,\infty} \tilde{\pi}_j(k) + \tilde{\pi}_{j,\infty}(I- \tilde{\pi}_j(k))
$$
$$
= \tilde{\pi}_{j,\infty} \tilde{\pi}_j(k) +
\tilde{\pi}_{j,\infty}(\tilde{h}_{\infty} -{\mathcal E}_j^+ -U_k)
(\tilde{h}(k) - {\mathcal E}_j^+)^{-1}(I- \tilde{\pi}_j(k)).
$$
Since $\tilde{\pi}_{j,\infty}(\tilde{h}_{\infty} -{\mathcal E}_j^+)=0$, we obtain the first equality in
\eqref{te4_0}. The second equality is obtained by taking the adjoint. In
 relations \eqref{te4_k} we have only exchanged the  role of $\tilde{h}(k)$ and $\tilde{h}_{\infty}$.
 \end{proof}
 Set
\bel{lau10}
\Phi_j(k) = \Phi_j(k;W) : = \left(\int_\re U_k(x) \tilde{\psi}_{j,\infty}(x)^2 dx\right)^{1/2}, \quad k \in \re.
\ee
By the dominated convergence theorem we have $\lim_{k \to +  \infty}
\Phi_j(k) = 0$. Note that
    \bel{te6}
    \Phi_j(k) = \left({\rm Tr}\,\tilde{\pi}_{j,\infty} U_k \tilde{\pi}_{j,\infty}\right)^{1/2} = \|\tilde{\pi}_{j,\infty}
    U_k^{1/2}\|_1 = \|U_k^{1/2} \tilde{\pi}_{j,\infty}\|_1 = \|\tilde{\pi}_{j,\infty}
    U_k^{1/2}\| = \|U_k^{1/2} \tilde{\pi}_{j,\infty}\|.
    \ee
\begin{follow} \label{tef1}
Fix $j \in {\mathbb N}$. Then we have
    \bel{te7}
    \|\tilde{\pi}_{j,\infty} - \tilde{\pi}_j(k)\|_1 = o(\Phi_j(k)), \quad k \to \infty.
    \ee
    \end{follow}
    \begin{proof}
    By \eqref{te4_0} and \eqref{te4_k} we have
    $$
     \tilde{\pi}_{j,\infty} - \tilde{\pi}_j(k) = - \tilde{\pi}_{j,\infty} U_k R_j^{\perp}({\mathcal E}_j^+) - R_{0,j}^{\perp}(E_{j})
    U_k \tilde{\pi}_{j,\infty} + R_{0,j}^{\perp}(E_{j})
    U_k (\tilde{\pi}_{j,\infty} - \tilde{\pi}_j(k)),
    $$
    i.e.
    $$
(I-R_{0,j}^{\perp}(E_{j})
    U_k ) (\tilde{\pi}_{j,\infty} - \tilde{\pi}_j(k)) = - \tilde{\pi}_{j,\infty} U_k R_j^{\perp}({\mathcal E}_j^+) -
R_{0,j}^{\perp}(E_{j})
    U_k \tilde{\pi}_{j,\infty}.
    $$
    Since $s-\lim_{k \to \infty} U_k = 0$ and the operator
    $R_{0,j}^{\perp}(E_{j})$ is compact and uniformly bounded, we have
    $\lim_{k \to \infty} \|R_{0,j}^{\perp}(E_{j}) U_k\| = 0$.
    Therefore, the operator $I-R_{0,j}^{\perp}(E_{j}) U_k$ is
    invertible for sufficiently great $k$, and for such $k$ we have
$$
\tilde{\pi}_{j,\infty} - \tilde{\pi}_j(k) = - (I-R_{0,j}^{\perp}(E_{j}) U_k)^{-1} (\tilde{\pi}_{j,\infty}
U_k R_j^{\perp}({\mathcal E}_j^+)
    +
R_{0,j}^{\perp}(E_{j})
    U_k \tilde{\pi}_{j,\infty}).
    $$
    Therefore,
    \bel{te8}
\|\tilde{\pi}_{j,\infty} - \tilde{\pi}_j(k)\|_1 \leq \|(I-R_{0,j}^{\perp}(E_{j}) U_k)^{-1}\| (
\|U_k^{1/2} R_j^{\perp}({\mathcal E}_j^+)\|
    +
\|R_{0,j}^{\perp}(E_{j})U_k^{1/2} \|)
     \|\tilde{\pi}_{j,\infty} U_k^{1/2}\|_1.
    \ee
    Arguing as above, we easily find that
    \bel{te9}
    \lim_{k \to \infty}\|U_k^{1/2} R_j^{\perp}({\mathcal E}_j^+)\|
    =
\lim_{k \to \infty}\|R_{0,j}^{\perp}(E_{j})U_k^{1/2} \| = 0.
    \ee
Now the combination of \eqref{te8}, \eqref{te9}, and \eqref{te6} implies
\eqref{te7}.
\end{proof}

\begin{pr} \label{tep2}
We have
    \bel{te10}
    {\mathcal E}_j^+ - E_j(k) = \Phi_j(k)^2( 1 + o(1)), \quad k \to \infty.
\ee
\end{pr}
\begin{proof}
Assume $k$ large enough. Evidently,
$$
{\mathcal E}_j^+ = {\rm Tr}\,\tilde{h}_{\infty} \tilde{\pi}_{j,\infty} = -\frac{1}{2\pi i} {\rm Tr}
\int_{\Gamma_j} \tilde{h}_{\infty} (\tilde{h}_{\infty} - \omega)^{-1} d\omega = -\frac{1}{2\pi
i} {\rm Tr} \int_{\Gamma_j} \omega (\tilde{h}_{\infty} - \omega)^{-1} d\omega
$$
where $\Gamma_j$ is a sufficiently small circle run over in the
anticlockwise direction which contains in its interior $E_j(k)$
and ${\mathcal E}_j^+$ but no other points from the spectra of
$\tilde{h}(k)$ and $\tilde{h}_{\infty}$. Similarly,
$$
E_{j}(k) =  -\frac{1}{2\pi i} {\rm Tr} \int_{\Gamma_j} \omega (\tilde{h}(k) -
\omega)^{-1} d\omega.
$$
Therefore,
    $$
    {\mathcal E}_j^+ - E_{j}(k) =  -\frac{1}{2\pi i} {\rm Tr} \int_{\Gamma_j} \omega \left((\tilde{h}_{\infty} -
\omega)^{-1} - (\tilde{h}(k) - \omega)^{-1}\right) d\omega =
$$
    \bel{te11}
 \frac{1}{2\pi i} {\rm
Tr} \int_{\Gamma_j} \omega (\tilde{h}_{\infty} - \omega)^{-1} U_k (\tilde{h}(k) -
\omega)^{-1} d\omega.
    \ee
Applying the Cauchy theorem, we easily get
    $$
    \frac{1}{2\pi i}
 \int_{\Gamma_j} \omega (\tilde{h}_{\infty} - \omega)^{-1} U_k (\tilde{h}(k) -
\omega)^{-1} d\omega =
    $$
    \bel{te12}
\tilde{\pi}_{j,\infty} U_k \tilde{\pi}_j(k) - {\mathcal E}_j^+\tilde{\pi}_{j,\infty} U_k R_j^{\perp}({\mathcal E}_j^+) -
    E_j R_{0,j}^{\perp}(E_{j}) U_k  \tilde{\pi}_j(k).
    \ee
    Comparing \eqref{te11} and \eqref{te12}, and bearing in mind
    \eqref{te6}, we obtain
$$
{\mathcal E}_j^+ - E_{j}(k) - \Phi_j(k)^2 =
$$
    \bel{te13}
{\rm Tr}\,\tilde{\pi}_{j,\infty} U_k (\tilde{\pi}_j(k) - \tilde{\pi}_{j,\infty}) - {\mathcal E}_j^+ {\rm Tr}\,\tilde{\pi}_{j,\infty}
U_k R_j^{\perp}({\mathcal E}_j^+) -
    E_j {\rm Tr}\, R_{0,j}^{\perp}(E_{j}) U_k  \tilde{\pi}_j(k).
    \ee
    In order to complete the proof of \eqref{te10}, it remains to
    show that the three terms on the r.h.s. of \eqref{te13} are of
    order $o(\Phi_j(k)^2)$ as $k \to \infty$. \\
    First, we have
    \bel{te14}
|{\rm Tr}\,\tilde{\pi}_{j,\infty} U_k (\tilde{\pi}_j(k) -
\tilde{\pi}_{j,\infty})| \leq \|\tilde{\pi}_{j,\infty} U_k^{1/2}\|
\|U_k^{1/2}\| \|\tilde{\pi}_j(k) - \tilde{\pi}_{j,\infty}\|_ 1 =
o(\Phi_j(k)^2), \quad k \to \infty,
    \ee
    by \eqref{te6}, \eqref{te7}, and the fact that $\|U_k^{1/2}\|$ is
    uniformly bounded with respect to $k \in \re$. \\
    Next, using  the trivial identities
$\tilde{\pi}_{j,\infty} = \tilde{\pi}_{j,\infty}^2$ and $R_j^{\perp}({\mathcal E}_j^+) \tilde{\pi}_j(k) = 0$,
as well as the cyclicity of the trace, we obtain
    \bel{te15}
      {\rm Tr}\,\tilde{\pi}_{j,\infty}
U_k R_j^{\perp}({\mathcal E}_j^+) =  -{\rm Tr}\,(\tilde{\pi}_j(k) - \tilde{\pi}_{j,\infty}) \tilde{\pi}_{j,\infty} U_k
R_j^{\perp}({\mathcal E}_j^+).
    \ee
    Therefore, similarly to \eqref{te14}, we have
    \bel{te16}
    | {\mathcal E}_j^+ {\rm Tr}\,\tilde{\pi}_{j,\infty}
U_k R_j^{\perp}({\mathcal E}_j^+)| \leq |{\mathcal E}_j^+| \|\tilde{\pi}_j(k) - \tilde{\pi}_{j,\infty}\|_1
\|\tilde{\pi}_{j,\infty} U_k^{1/2}\| \|U_k^{1/2}R_j^{\perp}({\mathcal E}_j^+)\|  = o(\Phi_j(k)^2)
    \ee
    as $k \to
\infty$. Finally, by analogy with \eqref{te15} we have
    $$
{\rm Tr}\, R_{0,j}^{\perp}(E_{j}) U_k  \tilde{\pi}_j(k)  = {\rm Tr}\,
R_{0,j}^{\perp}(E_{j}) U_k  \tilde{\pi}_j(k)  (\tilde{\pi}_j(k)  - \tilde{\pi}_{j,\infty}) =
$$
$$
{\rm Tr}\, R_{0,j}^{\perp}(E_{j}) U_k  \tilde{\pi}_{j,\infty} (\tilde{\pi}_j(k)  - \tilde{\pi}_{j,\infty}) +
{\rm Tr}\, R_{0,j}^{\perp}(E_{j}) U_k  (\tilde{\pi}_j(k) - \tilde{\pi}_{j,\infty})^2.
$$
Hence,
$$
|E_j {\rm Tr}\, R_{0,j}^{\perp}(E_{j}) U_k \tilde{\pi}_j(k)| \leq
$$
    $$
    |E_j(k)| \|R_{0,j}^{\perp}(E_{j}) U_k^{1/2}\| \|U_k^{1/2}
    \tilde{\pi}_{j,\infty}\| \|\tilde{\pi}_j(k)  - \tilde{\pi}_{j,\infty}\|_1 +
    $$
    \bel{te17}
    |E_j(k)| \|R_{0,j}^{\perp}(E_{j})
    U_k\| \|\tilde{\pi}_j(k) - \tilde{\pi}_{j,\infty}\|_1^2 = o(\Phi_j(k)^2), \quad k \to
\infty,
    \ee
    by \eqref{te6}, \eqref{te7}, and the fact that $|E_j(k)|$, $\|R_{0,j}^{\perp}(E_{j}(k)) U_k^{1/2}\|$,
    and $\|R_{0,j}^{\perp}(E_{j}(k)) U_k\|$ are uniformly bounded with respect to $k \in \re$.\\
Putting together \eqref{te13}, \eqref{te14}, \eqref{te16}, and
\eqref{te17}, we obtain \eqref{te10}.
\end{proof}
Now \eqref{te1} (and, hence, \eqref{lau6}) follows immediately from \eqref{te7} and \eqref{te10}.\\

\section{Effective Hamiltonians}
\label{section3} \setcounter{equation}{0}

Assume that $W$ is a non-decreasing function, and \eqref{5} holds
true. As explained in the introduction, for definiteness, we will
consider the case of positive perturbations, and respectively the
asymptotic behavior as $\lambda \downarrow 0$ of ${\mathcal
N}_j^+(\lambda)$, $j \in {\mathbb N}$, $\lambda \in (0,
2b + W_- - W_+)$. \\
Pick $j \in {\mathbb N}$, $A \in [-\infty, \infty)$ and $\lambda > 0$, and set
    $$
P_{j}(A) : = \int_{(A,{\infty})}^\oplus \pi_j(k) dk, \quad
{\mathcal P}_{j}(A) : = {\mathcal F}^* P_{j}(A) {\mathcal F},
\quad P_{j,\infty}(A) : = \int_{(A,{\infty})}^\oplus
\pi_{j,\infty}(k) dk,
$$
$$
T_{j}(\lambda; A) : = \int_{(A,{\infty})}^\oplus ({\mathcal E}_j^+
- E_j(k) + \lambda)^{-1/2}\pi_j(k) dk,
$$
$$
 T_{j,\infty}(\lambda; A) :
= \int_{(A,{\infty})}^\oplus ({\mathcal E}_j^+ - E_j(k) +
\lambda)^{-1/2}\pi_{j,\infty}(k) dk.
$$
 \begin{pr} \label{laup2} Assume $M \in L_0^{\infty}(\rd)$. Then the operator
$M{\mathcal F}^*T_{j,\infty}(\lambda; A) $ is compact for any $\lambda > 0$ and $A \in [-\infty,\infty)$. Moreover,
for any $A_1, A_2 \in [-\infty, \infty)$ the operator
\bel{lau16}
M{\mathcal F}^*(T_{j,\infty}(\lambda; A_1) - T_{j,\infty}(\lambda; A_2))
\ee
extends to a uniformly bounded and continuous operator for $\lambda \geq 0$.
\end{pr}
\begin{proof}
Denote by $\chi_R$ the characteristic function of a disk of radius
$R$ centered at the origin. For $\lambda > 0$ and $A \in [-\infty,
\infty)$ write
    \bel{lau17}
    M{\mathcal F}^*T_{j,\infty}(\lambda; A)
= \chi_R M{\mathcal F}^*T_{j,\infty}(\lambda; A) +  (1-\chi_R)
M{\mathcal F}^*T_{j,\infty}(\lambda; A).
    \ee
 The first operator at the r.h.s of
\eqref{lau17} is Hilbert-Schmidt for any $R \in (0,\infty)$, and
the norm of the second one tends to zero as $R \to \infty$. Hence,
the operator $M{\mathcal F}^*T_{j,\infty}(\lambda; A)$ is compact.
 Further, the case $A_1 = A_2$ in \eqref{lau16} is trivial so
that we suppose $A_1 \neq A_2$. Define the value for $\lambda = 0$
of the operator in \eqref{lau16} as
$$
M{\mathcal F}^* \int_{(A-,A_+)}^\oplus ({\mathcal E}_j^+
- E_j(k))^{-1/2} \pi_{j,\infty}(k) dk
$$
with $A_- : = \min\{A_1,A_2\}$ and $A_+ : = \max\{A_1,A_2\}$. Now the uniform boundedness
for $\lambda \geq 0$ of the operator
in \eqref{lau16} follows from the estimates
$$
\|M{\mathcal F}^*(T_{j,\infty}(\lambda; A_1) -T_{j,\infty}(\lambda; A_2))\|
\leq \|M\|_{L^{\infty}(\rd)} \sup_{k \in (A_-, A_+]} ({\mathcal E}_j^+
- E_j(k))^{-1/2}, \quad \lambda \geq 0,
$$
while the uniform continuity of this operator for $\lambda \geq 0$ follows from the estimates
$$
\|M{\mathcal F}^*((T_{j,\infty}(\lambda_1; A_1) - T_{j,\infty}(\lambda_1; A_2)) - (T_{j,\infty}(\lambda_2; A_1) -T_{j,\infty}(\lambda_2; A_2)))\| \leq
$$
$$
|\lambda_1 - \lambda_2| \|M\|_{L^{\infty}(\rd)} \sup_{k \in (A_-,
A_+]} ({\mathcal E}_j^+ - E_j(k))^{-2}, \quad \lambda_1, \lambda_2
\geq 0.
$$
\end{proof}
Let $s>0$ and $T=T^*$ be a compact operator acting in a given
Hilbert space. Set
 $$
n_{\pm}(s; T) : = {\rm rank}\,{\mathbb P}_{(s,\infty)}(\pm T).
$$
In the case where $T$ is  linear and compact but not necessarily
self-adjoint  (in particular, $T$ could act between two different
Hilbert spaces), we will use also the notations
$$
n_*(s; T) : = n_+(s^2; T^* T), \quad s> 0.
$$
Of course, we have $n_*(s; T) = n_*(s; T^*)$ and hence $n_+(s; T^*T) = n_+(s; T T^*)$ for any $s>0$.
For further references we recall here the well-known Weyl
inequalities
    \bel{lau11}
    n_+(s_1 + s_2; T_1 + T_2) \leq n_+(s_1; T_1) + n_+(s_2; T_2)
    \ee
    where $s_j > 0$ and $T_j$, $j=1,2$, are
linear self-adjoint operators acting in a given Hilbert space. In
the case where $T_1$ and $T_2$ are linear and compact but not
necessarily self-adjoint, we recall also the Ky-Fan inequalities
\bel{lau13} n_*(s_1 + s_2; T_1 + T_2) \leq n_*(s_1; T_1) + n_*(s_2;
T_2), \quad s_1, s_2 > 0. \ee
\begin{theorem} \label{lauth1} Assume that $V \in L_0^{\infty}(\rd; \re)$. Fix $j \in {\mathbb N}$
and $A \in [-\infty, \infty)$. Then for any $\varepsilon \in
(0,1)$ we have
$$
n_+(1+\varepsilon; T_{j,\infty}(\lambda;A){\mathcal F} V {\mathcal
F}^*T_{j,\infty}(\lambda;A)) + O(1) \leq
$$
$$
{\mathcal N}_j^+(\lambda) \leq
$$
\bel{lau30} n_+(1-\varepsilon; T_{j,\infty}(\lambda;A){\mathcal F}
V {\mathcal F}^*T_{j,\infty}(\lambda;A)) + O(1), \quad \lambda
\downarrow 0. \ee
\end{theorem}
\begin{proof}
The Birman-Schwinger principle implies
    \bel{57}
{\mathcal N}_j^+(\lambda) = n_-(1; V^{1/2} (H_0 - {\mathcal E}_j^+
- \lambda)^{-1} V^{1/2}) + O(1), \quad \lambda \downarrow 0.
    \ee
    Pick $\tilde{A} \in \re$. Applying the Weyl inequalities \eqref{lau11}, we get
    $$
n_+(1+s; V^{1/2} ({\mathcal E}_j^+ - H_0 + \lambda)^{-1} {\mathcal
P}_{j,\tilde{A}} V^{1/2}) - n_-(s; V^{1/2} ({\mathcal E}_j^+ - H_0
+ \lambda)^{-1}(I - {\mathcal P}_{j,\tilde{A}}) V^{1/2}) \leq
$$
$$
n_-(1; V^{1/2} (H_0 - {\mathcal E}_j^+ -
\lambda)^{-1} V^{1/2}) \leq
$$
    \bel{59}
n_+(1-s; V^{1/2} ({\mathcal E}_j^+ - H_0 + \lambda)^{-1} {\mathcal
P}_{j,\tilde{A}} V^{1/2}) + n_+(s; V^{1/2} ({\mathcal E}_j^+ - H_0
+ \lambda)^{-1}(I - {\mathcal P}_{j,\tilde{A}}) V^{1/2}),
    \ee
for any $s \in (0,1)$. By $V \in L_0^{\infty}(\rd; \re)$ and the
diamagnetic inequality, we easily find that
 \bel{lau31}
 n_{\pm}(s; V^{1/2} ({\mathcal E}_j^+ - H_0 + \lambda)^{-1}(I - {\mathcal
P}_{j,\tilde{A}}) V^{1/2}) = O(1), \quad \lambda \downarrow 0. \ee
Further, for any $r>0$ we have
$$
n_+(r^2; V^{1/2} ({\mathcal E}_j^+ - H_0 + \lambda)^{-1}{\mathcal
P}_{j,\tilde{A}} V^{1/2}) =
$$
$$
n_+(r^2; V^{1/2} {\mathcal F}^*
\int^{\oplus}_{(\tilde{A},\infty)}({\mathcal E}_j^+ - E_j(k) +
\lambda)^{-1} \pi_j(k) dk {\mathcal F} V^{1/2}) =
$$
    \bel{lau32}
n_*(r; V^{1/2} {\mathcal F}^* T_j(\lambda; \tilde{A})).
    \ee
    Applying the Ky-Fan inequalities \eqref{lau13}, we obtain
$$
n_*(r(1+s); V^{1/2} {\mathcal F}^* T_{j,\infty}(\lambda;
\tilde{A})) - n_*(rs; V^{1/2} {\mathcal F}^*(
T_{j,\infty}(\lambda; \tilde{A}) -  T_j(\lambda; \tilde{A}))) \leq
$$
$$
n_*(r; V^{1/2} {\mathcal F}^* T_j(\lambda; \tilde{A})) \leq
 $$
\bel{lau33} n_*(r(1-s); V^{1/2} {\mathcal F}^*
T_{j,\infty}(\lambda; \tilde{A})) + n_*(rs; V^{1/2} {\mathcal
F}^*( T_{j,\infty}(\lambda; \tilde{A}) -  T_j(\lambda;
\tilde{A}))). \ee
Now note that
$$
\| V^{1/2} {\mathcal F}^*( T_{j,\infty}(\lambda; \tilde{A}) -
T_j(\lambda; \tilde{A}))\| \leq
$$
    \bel{lau34}
\|V\|_{L^{\infty}(\rd)}^{1/2} \sup_{k > \tilde{A}} ({\mathcal
E}^+_j - E_j(k))^{-1/2} \|\pi_j(k) - \pi_{j,\infty}(k)\|,
    \ee
uniformly with respect to $\lambda > 0$.
 By \eqref{lau34} and Theorem \ref{teth1} we find that for each
$q>0$ there exists $A_0 = A_0(q)$ such that $\tilde{A} \geq  A_0(q)$
implies
$$
\| V^{1/2} {\mathcal F}^*( T_{j,\infty}(\lambda; \tilde{A}) -
T_j(\lambda; \tilde{A}))\| \leq q
$$
for each $\lambda > 0$. Choosing $\tilde{A} \geq A_0(rs)$ in
\eqref{lau33} we find then that \bel{lau35} n_*(rs; V^{1/2}
{\mathcal F}^*( T_{j,\infty}(\lambda; \tilde{A}) -  T_j(\lambda;
\tilde{A}))) = 0 \ee for each $\lambda > 0$. Next, the Ky-Fan
inequalities \eqref{lau13} imply that for any $\lambda > 0$,
$r>0$, $s \in (0,1)$ and $A$,
  $\tilde{A}$, we have
$$
n_*(r (1+s); V^{1/2} {\mathcal F}^* T_{j,\infty}(\lambda; A)) -
n_*(rs; V^{1/2} {\mathcal F}^*(T_{j,\infty}(\lambda; A) -
T_{j,\infty}(\lambda; \tilde{A})){)} \leq
$$
$$
n_*(r; V^{1/2} {\mathcal F}^* T_{j,\infty}(\lambda; \tilde{A})) \leq
$$
    \bel{lau36}
    n_*(r (1-s); V^{1/2} {\mathcal F}^*
T_{j,\infty}(\lambda; A)) + n_*(rs; V^{1/2} {\mathcal
F}^*(T_{j,\infty}(\lambda; A) - T_{j,\infty}(\lambda;
\tilde{A}))).
    \ee
    By Proposition \ref{laup2} we have
    \bel{lau37}
    n_*(r; V^{1/2} {\mathcal F}^*(T_{j,\infty}(\lambda; A)
- T_{j,\infty}(\lambda; \tilde{A}))) = O(1), \quad \lambda
\downarrow 0,
    \ee
    for any fixed $r>0$. Finally we note that
\bel{lau38} n_*(r; V^{1/2} {\mathcal F}^* T_{j,\infty}(\lambda;
A)) = n_+(r^2; T_{j,\infty}(\lambda; A){\mathcal F} V {\mathcal
F}^*T_{j,\infty}(\lambda; A)) \ee for each $r>0$, $\lambda > 0$.
Putting together \eqref{57} -- \eqref{lau33}, and \eqref{lau35} --
\eqref{lau38}, we obtain \eqref{lau30}.
\end{proof}
 Note that the
    operator $T_{j,\infty}(\lambda; A) {\mathcal F} V {\mathcal F}^*
T_{j,\infty}(\lambda; A)$ regarded as an operator on the Hilbert space
$P_{j,\infty}(A) L^2(\rd)$, is unitarily equivalent to the operator $S_j(\lambda; A)^* S_j(\lambda; A)$ where
$S_j(\lambda; A) : L^2(A,\infty) \to L^2(\rd)$ is an operator with integral kernel
\bel{defsj}
(2\pi)^{-1/2}  V(x,y)^{1/2} e^{iky} \psi_{j,\infty}(x;k) ({\mathcal E}_j^+ - E_j(k) + \lambda)^{-1/2}, \quad k \in \re, \quad (x,y) \in \rd, \quad \lambda > 0,
\ee
the function $\psi_{j,\infty}$ being defined in \eqref{sof9}.
Therefore, \bel{lau60} n_+(r; T_{j,\infty}(\lambda; A) {\mathcal
F} V {\mathcal F}^* T_{j,\infty}(\lambda; A)) = n_+(r;
S_j(\lambda; A)^* S_j(\lambda; A)), \quad r>0, \quad \lambda > 0.
    \ee
    For $(x,\xi) \in T^*\re = \rd$ and $j \in \N$ set
$$
\Psi_{x,\xi;j}(k) = e^{-ik\xi}\psi_j(x;k), \quad k \in \re.
$$
 Note that for each $j \in \N$ the system
$\left\{\Psi_{x,\xi;j}\right\}_{(x,\xi) \in T^* \re}$ is
overcomplete with respect to the measure $\frac{1}{2\pi} dx d\xi$,
i.e. for each $f \in L^2(\re)$ we have
$$
\frac{1}{2\pi} \int_{T^*\re} |\langle f, \Psi_{x,\xi;j} \rangle
|^2 dx d\xi = \int_\re |f(k)|^2 dk
$$
 (see \cite{bershu} or
\cite[Section 24]{shu}). For $(x,\xi) \in T^*\re$ and $j \in \N$
set $ P_{x,\xi;j} : = \langle \cdot, \Psi_{x,\xi;j}\rangle
\Psi_{x,\xi;j}$, and introduce the operator
$$
{\mathcal V}_j = \frac{1}{2\pi} \int_{T^*\re} V(x,\xi) P_{x,\xi;
j} dx d\xi
$$
where the integral is understood in the weak sense. Then
${\mathcal V}_j$ can be interpreted as {a} $\Psi$DO with
contravariant (generalized anti-Wick symbol) $V$ (see
\cite{bershu}). Moreover, we have
    \bel{jun100}
S_j(\lambda; -\infty)^* S_j(\lambda; -\infty) = ({\mathcal E}_j^+
- E_j + \lambda)^{-1/2} {\mathcal V}_j ({\mathcal E}_j^+ - E_j +
\lambda)^{-1/2}.
    \ee
Bearing in mind \eqref{lau60} and \eqref{jun100}, and {  applying} the
Birman-Schwinger principle, we find that \eqref{lau30} with $A =
-\infty$ and $\varepsilon \in (0,1)$ can be re-written as
$$
{\rm rank}\,{\mathbb P}_{({\mathcal E}_j^+ +
\lambda,\infty)}({{E}_j} + (1+\varepsilon)^{-1} {\mathcal
V}_j) + O(1) \leq
$$
$$
{\mathcal N}_j^+(\lambda) \leq
$$
$$
{\rm rank}\,{\mathbb P}_{({\mathcal E}_j^+ + \lambda,
\infty)}({{E}_j}  + (1-\varepsilon)^{-1} {\mathcal V}_j) +
O(1), \quad \lambda \downarrow 0.
$$
Thus, the operator ${{E}_j} + {\mathcal V}_j$ could be
interpreted as the effective Hamiltonian which governs the
asymptotics of ${\mathcal N}_j^+(\lambda)$ as $\lambda \downarrow
0$, the multiplier by $E_j$ being its ``kinetic" part, and the
$\Psi$DO ${\mathcal V}_j$ being its ``potential" part.

\section{Sufficient condition for the boundedness of ${\mathcal N}_j^+(\lambda)$}
\label{section4}
\setcounter{equation}{0}

Assume that $W$ is a bounded non-decreasing function with $W_- <
W_+$.  Set
\bel{lau24}
x^+ : = \inf \{x \in \re \, | \, W(x) =
W_+\}.
    \ee
    By the assumption $W_- < W_+$, we have $x^+ > -\infty$. Set
$$
{\mathcal X}: = \left\{x \in \re \, | \, \mbox{there exists} \; y
\in \re \; \mbox{such that}\;  (x,y) \in {\rm ess}\,{\rm
supp}\,V\right\},
$$
$$
X^- : = \inf {\mathcal X}, \quad  X^+ : = \sup {\mathcal X}.
$$
 \begin{theorem} \label{repr1}
 Let $W$ be a bounded and non-decreasing function with $W_- < W_+$ and $x^+ \leq \infty$.
    Assume that $V \in L_0^{\infty}(\rd)$, $V \geq 0$, $-\infty < X^- < X^+ < \infty$. Suppose in addition that
$ {\rm ess}\,\sup_{x \in \re} \int_{\re} V(x,y) dy < \infty$, and
    \bel{lau50}
    X^+  < x^+.
    \ee
     Then we have
    \bel{rened3}
    {\mathcal N}_j^+(\lambda) = O(1), \quad \lambda \downarrow 0, \quad {   j \in } {\mathbb N}.
    \ee
    \end{theorem}
    In order to prove Theorem \ref{repr1} we need some information
    on the asymptotic behaviour as $k \to \infty$ of the function $\Phi_j(k)$
    defined in \eqref{lau10} which by Proposition \ref{tep2}
    determines the asymptotics of ${\mathcal E}_j - E_j(k)$.
    Let $w_-, w_+ \in \re$, $w_- < w_+$, $x_0 \in \re$. Put
    \bel{sof7}
    {\bf w}(x) : = \left\{
    \begin{array} {l}
    w_+ \quad {\rm if} \quad x \geq x_0,\\
    w_- \quad {\rm if} \quad x < x_0.
    \end{array}
    \right.
    \ee
\begin{pr} \label{laup1}
Assume that  $w_- < w_+$. Then we have
     \bel{lau25}
\Phi_j(k; {\bf w})^2 = \frac{(w_+ - w_-)}{2} p_j k^{2j-3}
e^{-(b^{-1/2}k - b^{1/2}x_0)^2}(1 + o(1)), \quad k \to \infty, \ee
the number $p_j$ being defined in \eqref{sof5}.
\end{pr}
We omit the simple proof of the proposition, based on the standard
Laplace method for approximate evaluation of integrals depending
on a large parameter.
    \begin{proof}[Proof of Theorem \ref{repr1}]
By the upper bound in \eqref{lau30}, and \eqref{lau60}, it
suffices to show that \bel{lau61} n_*(r; S_j(\lambda; A)) = O(1),
\quad \lambda \downarrow 0, \ee for any fixed $r>0$ and  $A \in
[-\infty, \infty)$. We have \bel{lau62} n_*(r; S_j(\lambda; A))
\leq r^{-2} {\rm Tr}\,  S_j(\lambda; A)^* S_j(\lambda; A) =
\frac{1}{2\pi r^2} {\mathcal I}_0(\lambda) \ee where
$$
{\mathcal I}_0(\lambda) : = \int_A^{\infty} \int_{\rd} ({\mathcal E}_j^+ - E_j(k; b,W) +
\lambda)^{-1} \psi_{j,\infty}(x;k)^2 V(x,y) dx dy dk.
$$
Now pick $\tilde{x} \in (X^+, x^+)$ which is possible due to
\eqref{lau50}, and set
    \bel{fin1}
\tilde{W}(x) = \left\{
\begin{array} {l}
W_+ \quad {\rm if} \quad x \geq \tilde{x},\\
W(\tilde{x}) \quad {\rm if} \quad x < \tilde{x}.
\end{array}
\right.
 \ee
Since $W(x) \leq \tilde{W}(x)$, $x \in \re$, {the} mini-max
principle implies
$$
({\mathcal E}_j^+ - E_j(k; b,W) + \lambda)^{-1} \leq ({\mathcal E}_j^+ - E_j(k; b, \tilde{W}) + \lambda)^{-1}, \quad k \in \re, \quad j \in {\mathbb N}, \quad \lambda > 0.
$$
Therefore,
\bel{lau63}
{\mathcal I}_0(\lambda) \leq {\Big( {\rm ess}\,\sup_{x \in \re} \int_{\re} V(x,y) dy \Big)}\; {\mathcal I}_1(\lambda)
\ee
where
$$
{\mathcal I}_1(\lambda) : = \int_A^{\infty} \int_{X^-}^{X^+}
({\mathcal E}_j^+ - E_j(k; b,\tilde{W}) + \lambda)^{-1}
\psi_{j,\infty}(x;k)^2 dx  dk.
$$
Taking into account \eqref{sof6}, \eqref{te10}, and \eqref{lau25},
and bearing in mind that the interval  $[X^-, X^+]$ is compact, we
find that for sufficiently large $A>0$ and any $\lambda \geq 0$ we
have
$$
{\mathcal I}_1(\lambda) \leq 4(W_+ - W(\tilde{x}))^{-1} \max_{x
\in  [X^-, X^+]}e^{-b(x^2 - \tilde{x}^2)} \int_A^{\infty}
\int_{X^-}^{X^+} k e^{-2k(\tilde{x} - x)} dx dk
$$
    \bel{lau64}
    \leq 2(W_+ - W(\tilde{x}))^{-1} \max_{x \in  [X^-,
X^+]}e^{-b(x^2 - \tilde{x}^2)} \int_A^{\infty} e^{-2k(\tilde{x} -
X^+)} dk  < \infty,
    \ee
    due to $\tilde{x} >
X^+$. Putting together \eqref{lau62} - \eqref{lau64}, we obtain
\eqref{lau61}, and hence, \eqref{rened3}.
\end{proof}

\section{Reduction to operators {in} holomorphic  spaces}
\label{section5} \setcounter{equation}{0} In what follows we
assume that there exist bounded domains $\Omega_\pm \subset \rd$
with Lipschitz boundaries, and constants $c_0^{\pm} > 0$ such that
    \bel{jun2}
    c_0^- \chi_{\Omega_-}(x,y) \leq V(x,y) \leq  c_0^+ \chi_{\Omega_+}(x,y), \quad (x,y) \in \rd,
    \ee
    where $\chi_{\Omega_\pm}$ denotes the characteristic function of the domain $\Omega_\pm$.
    Next, for $\delta \in (0,1/2)$ introduce the intervals
    $$
    I_- = I_-(\delta) : = (\delta, 1-\delta), \quad I_+ = I_+(\delta) : = (0, 1+\delta).
    $$
    In what follows we will assume that the infimum $x^+$ defined in {\eqref{lau24}} satisfies $x^+ < \infty$
    because in the case $x^+ = \infty$ Theorem \ref{repr1} implies {  ${\mathcal N}^+_j(\lambda) = O(1)$}
    as $\lambda \downarrow 0$. Since the operator $H_0$ is invariant under magnetic translations,
    we will assume that $x^+ = 0$ without any loss of generality. \\
    Let $\delta \in (0,1/2)$ and $m>0$. Define the operator $\Gamma^-_\delta(m): L^2(I_-) \to L^2(\Omega_-)$
    as the operator with integral kernel
    $$
    \pi^{-1/2} m  e^{-bx^2/2} e^{m(x+iy)k} k^{1/2}, \quad k \in I_-, \quad (x,y) \in \Omega_-,
    $$
    and the operator $\Gamma^+_\delta(m): L^2(I_+) \to L^2(\Omega_+)$ as the operator with integral kernel
    $$
    \pi^{-1/2} m  e^{-bx^2/2} e^{m(x+iy + \delta)k} k^{1/2}, \quad k \in I_+, \quad (x,y) \in \Omega_+.
    $$
    {\em Remark}: Introduce the set
          \bel{jul20}
          {\mathcal B}(\Omega_\pm) : = \left\{u \in L^2(\Omega_\pm) \, | \, u \, \mbox{is analytic in} \, \Omega_\pm\right\}
          \ee
          considered as a subspace of the Hilbert space $L^2(\Omega_\pm; e^{-bx^2} dxdy)$. Note that as a functional set ${\mathcal B}(\Omega_\pm)$ coincides with the Bergman space over
          $\Omega_\pm$ (see e.g. \cite[Subsection 3.1]{ha}). Then,
    up to unitary equivalence, the operators $\Gamma_\delta^{\pm}(m)$ map $L^2(I_\pm)$ into ${\mathcal B}(\Omega_\pm)$.

\begin{theorem} \label{lauth2}
Suppose that $W$ is a bounded non-decreasing function with $W_- <
W_+$. Assume that  $V \in L_0^{\infty}(\rd; \re)$  satisfies
\eqref{jun2}. Then we have
    \bel{jun3a}
    n_*(r; S_j(\lambda; A)) \geq n_*(r (1 + \varepsilon) \sqrt{(W_+ - W_-)/c_0^-}; \Gamma^-_\delta (\sqrt{b |{\ln{\lambda}}|})) + O(1),
    \ee
    \bel{jun3}
    n_*(r; S_j(\lambda; A)) \leq n_*(r (1 - \varepsilon) \sqrt{(W_+ - W(-\delta))/c_0^+} e^{-b\delta^2/2}; \Gamma^+_\delta (\sqrt{b |{\ln{\lambda}}|})) + O(1),
    \ee
as $\lambda \downarrow 0$, for all $j \in \N$, $A>0$, $\varepsilon
\in (0,1)$, $\delta \in (0,1/2)$ and $r>0$.
\end{theorem}

We will divide the proof of Theorem \ref{lauth2} into two
propositions.

 Define the non-decreasing functions
       $$
       W_0^-(x) : = \left\{
       \begin{array} {l}
       W_+ \quad \mbox{if} \quad x > 0,\\
       W_- \quad \mbox{if} \quad x \leq 0,
       \end{array}
       \right.
       $$
       $$
       W_0^+(x) = W_0^+(x; \delta) : = \left\{
       \begin{array} {l}
       W_+ \quad \mbox{if} \quad x \geq -\delta,\\
       W(-\delta) \quad \mbox{if} \quad x < -\delta,
       \end{array}
       \right. \quad \delta > 0.
       $$
       Since $x^+ = 0$ and $\delta > 0$ we have
       \bel{jun3b}
       W_0^-(x) \leq W(x) \leq W_0^+(x; \delta), \quad x \in \re.
       \ee
       Set
       \bel{aug1}
\omega_\pm: = \left\{x \in \re \, | \, \mbox{there exists} \; y
\in \re \; \mbox{such that}\;  (x,y) \in \Omega_\pm \right\}.
    \ee
Let $\lambda > 0$, $A \in [-\infty, \infty)$. Fix $j \in {\mathbb
N}$. Define  $Q_j^{\pm}(\lambda;A) : L^2(A,\infty) \to
L^2(\Omega_\pm)$ as the operator with integral kernel
    \bel{apr1}
    \left(\frac{p_j}{2\pi}\right)^{1/2} e^{iky} e^{-(b^{1/2}x-b^{-1/2}k)^2/2}
    \left({\mathcal E}_j^+ - E_j(k; W_0^\pm) + \lambda\right)^{-1/2} (-k)^{j-1},
    \ee
   with $k \in (A,\infty)$,
    $(x,y) \in \Omega_\pm$, the number $p_j$ being defined in \eqref{sof5}. {For $S_j(\lambda;A)$ defined by \eqref{defsj}, we have:}
    \begin{pr} \label{sanpr1}
    Assume that $W$ and $V$ satisfy the assumptions of Theorem \ref{lauth2}. Then for every $A>0$, $r>0$,
    and $\varepsilon \in (0,1)$, we have
    \bel{sof19}
    n_*(r(1+\varepsilon); \sqrt{c_0^-} Q_j^-(\lambda;A)) + O(1) \leq n_*(r; S_j(\lambda;A)) \leq
    n_*(r(1{- \varepsilon}); \sqrt{c_0^+} Q_j^+(\lambda;A)) + O(1),
    \ee
    as $\lambda \downarrow 0$, $c_0^\pm$ being the constants occurring in \eqref{jun2}.
    \end{pr}
    \begin{proof}
   Inequalities \eqref{jun2} and \eqref{jun3b}, combined with the mini-max
    principle, imply the  estimates
        \bel{sof20}
         n_*(r; \sqrt{c_0^-} \tilde{S}_j^-(\lambda;A)) \leq n_*(r; S_j(\lambda;A)) \leq
         n_*(r; \sqrt{c_0^+} \tilde{S}_j^+(\lambda;A))
         \ee
         where $\tilde{S}_j^{\pm}(\lambda;A) : L^2(A,\infty) \to L^2(\Omega_\pm)$ is the operator with integral kernel
         $$
         (2\pi)^{-1/2} e^{iky} \psi_{j,\infty}(x;k) ({\mathcal E}_j^+ - E_j(k; W_0^\pm) + \lambda)^{-1/2},
         \quad k \in \re, \quad (x,y) \in \Omega_\pm.
         $$
         In the case $j=1$ inequality \eqref{sof20} yields immediately \eqref{sof19} since in this case we have
         $\tilde{S}_1^\pm(\lambda;A) = Q_1^\pm(\lambda;A)$. Assume $j \geq 2$. Then we have
         \bel{fin4}
        \psi_{j,\infty}(x;k) =  p_j^{1/2}  \sum_{l=0}^{j-1} P_{l,j}(x) (-k)^{j-1-l}e^{-(b^{-1/2}k - b^{1/2}x)^2/2}, \quad x \in \re, \quad k \in \re,
         \ee
         where $P_{l,j}$ is a polynomial of grade less than or equal to $l$, and $P_{0,j} = 1$.
         Therefore,
         $$
         \tilde{S}_j^\pm(\lambda;A) =  \sum_{l=0}^{j-1} P_{l,j} Q_j^\pm(\lambda;A) B^{l}
         $$
         where the operator $B : L^2(A,\infty) \to L^2(A,\infty)$ with $A>0$ is defined by
         $$
         (Bu)(k) = k^{-1} u(k), \quad k \in (A,\infty), \quad u \in L^2(A,\infty).
         $$
         Further, for each $\eta \in (0,1)$,  the elementary
         inequalities
         $$
         \tilde{S}_j^-(\lambda;A)^* \tilde{S}_j^-(\lambda;A) \geq
            $$
         \bel{sof22}
         (1 - \eta) Q_j^-(\lambda;A)^* Q_j^-(\lambda;A) -  (j-1)(\eta^{-1} - 1) c_2^- \sum_{l=1}^{j-1} B^{l} Q_j^-(\lambda;A)^* Q_j^-(\lambda;A) B^{l}
         \ee
         $$
         \tilde{S}_j^+(\lambda;A)^* \tilde{S}_j^+(\lambda;A) \leq
            $$
         \bel{sof22a}
         (1 + \eta) Q_j^+(\lambda;A)^* Q_j^+(\lambda;A) +  (j-1) (\eta^{-1} + 1)
         c_2^+
         \sum_{l=1}^{j-1} B^{l} Q_j^+(\lambda;A)^* Q_j^+(\lambda;A) B^{l}
         \ee
         hold for $c_2^\pm : = \max_{l=1,\ldots, j-1} \sup_{x \in \omega_\pm} P_{l,j}(x)^2$.
          Let us consider now {the} quadratic forms
        $$
         a_l^\pm[u] = a_l^\pm[u; \lambda, j] : =
    $$
     \bel{sof21}
         \int_{\Omega_\pm}e^{-bx^2}\left|\int_A^{\infty}  e^{k(x+iy)} e^{-b^{-1}k^2/2}
         \left({\mathcal E}_j^+ - E_j(k; W_0^\pm) + \lambda\right)^{-1/2} (-k)^{j-l-1}u(k)dk\right|^2dxdy
         \ee
         with $u \in C_0^{\infty}(A,\infty)$, $\lambda > 0$, $j \geq 2$, $l=0,\ldots,j-2$. Evidently,
         $a_l^\pm[u] \geq 0$, and
    $a_l^\pm[u] = 0$ implies $u = 0$. Denote by $D[a_l^\pm]$,
    $l=0,\ldots,j - 2$, the completion of $C_0^{\infty}(A,\infty)$
    in the norm generated by $a_l^\pm$. \\
         If
         we can prove now that the quadratic form $a_{l+1}^\pm[u]$, $u \in C_0^{\infty}(A,\infty)$,
         is closable in {$D[a_l^\pm]$}, and the operator ${\mathbb A}_l^\pm$ generated by the closure of {$a_{l+1}^\pm$}
         is compact in $D[a_l^\pm]$, and $\sigma({\mathbb A}_l^\pm)$ is
         independent of $\lambda$, then it would follow easily from
         \eqref{sof22} - \eqref{sof22a} that for each $\varepsilon \in (0,1)$ there exist
         subspaces
         ${\mathcal H}_\pm$ of $C_0^{\infty}(A,\infty)$ such that the codimensions ${\rm codim}\,{{\mathcal
         H_\pm}}$ are finite and independent of $\lambda$, and
         \bel{jul11}
         \|\tilde{S}_j^-(\lambda;A)u\|^2 \geq (1+\varepsilon)^{-2} \|Q_j^-(\lambda;A)u\|^2, \quad u \in {\mathcal
         H_-},
         \ee
         \bel{jul12}
         \|\tilde{S}_j^+(\lambda;A)u\|^2 \leq (1-\varepsilon)^{-2} \|Q_j^+(\lambda;A)u\|^2, \quad u \in {\mathcal
         H_+}.
         \ee
         Combining \eqref{jul11} - \eqref{jul12} with standard variational arguments (see. e.g. \cite[Lemma 1.13]{birsol} and the proof of \cite[Lemma 1.16]{birsol}), we get
         \bel{san10}
         n_*(r; \tilde{S}_j^-(\lambda;A)) \geq n_*(r(1+\varepsilon); Q_j^-(\lambda; A)) - {\rm codim}\,{\mathcal
         H_-},
         \ee
         \bel{san10a}
         n_*(r; \tilde{S}_j^+(\lambda;A)) \leq n_*(r(1-\varepsilon); Q_j^+(\lambda; A)) + {\rm codim}\,{\mathcal
         H_+}.
         \ee
         Putting together  \eqref{san10} -  \eqref{san10a} and \eqref{sof20}, we arrive at
         \eqref{sof19}.\\
         Let us now prove the necessary properties of  ${\mathbb
         A}_l^\pm$.  Define the
         operators $F_l$ by
         $$
         (F_l v)(z) : = \int_{\re} e^{zk} k^l v(k) dk, \quad
         v \in C_0^{\infty}(\re), \quad z \in \ce.
         $$
         Note that  $F_l v$, $l=0,1$, are
         entire functions in $\ce$, and $(F_1 v)(z) = \frac{\partial (F_0 v)}{\partial z}(z)$.
Moreover, the operators $F_l$, $l=0,1$, can be extended as
continuous operators from ${\mathcal D}'_{\rm comp}(\re)$, the
space of distributions with compact support dual to
$C_0^{\infty}(\re)$, into the space of functions entire in $\ce$.
         Set
         $$
         f_l^\pm[v] : = \int_{\Omega_\pm} e^{-bx^2}|(F_l v)(x+iy) |^2 dx
         dy, \quad v \in C_0^{\infty}(A,\infty), \quad l=0,1.
         $$
         Further, for $j \geq 2$, $l
          = 0,\ldots, j-2$, and $\lambda > 0$, define the operator ${\mathcal U}_{j,l, \lambda}$ by
         $$
           ({\mathcal U}^\pm_{j, l, \lambda}u)(k) : =
          {e^{-b^{-1}k^2/2}} \left({\mathcal E}_j^+ - E_j(k; W_0^\pm) + \lambda\right)^{-1/2} k^{j-l-2}u(k), \quad k \in (A, \infty).
          $$
          Note that the mapping ${\mathcal U}^\pm_{j, l, \lambda}
          : C_0^{\infty}(A,\infty) \to C_0^{\infty}(A,\infty)$ is
          bijective, and we have
          \bel{jul7}
          a_l^\pm[u] = f_1^\pm[{\mathcal U}^\pm_{j, l,
          \lambda}u], \quad a_{l+1}^\pm[u] = f_0^\pm[{\mathcal U}^\pm_{j, l,
          \lambda}u], \quad u \in C_0^{\infty}(A,\infty), \quad l
          = 0, \ldots, j-2.
          \ee
          Denote by $D[f_1^\pm]$ the closure of
          $C_0^{\infty}(A,\infty)$ in the norm generated by the
          quadratic form $f_1^\pm$.
          If we can prove now that the quadratic form $f_0^\pm$
          is closable in $D[f_1^\pm]$, and the operator
          ${\mathbb F}^\pm$ generated by its closure, is compact in
          $D[f_1^\pm]$, then \eqref{jul7} would imply that
          $a_{l+1}^\pm$ is closable in $D[a_{l}^\pm]$, $l = 0,\ldots, j-2$, the
          operator ${\mathbb A}_l^{\pm}$ is compact, and its spectrum $\sigma({\mathbb
          A}_l^{\pm})$ is independent of $\lambda$. \\
          Let us now prove the necessary properties of  ${\mathbb F}^\pm$. Consider first $D[f_1^\pm
          + f_0^\pm]$, the closure of $C_0^\infty(A,\infty)$ in
          the norm generated by the quadratic form $f_1^\pm
          + f_0^\pm$. The quadratic form $f_0^\pm$ is
          bounded, and hence closable in $D[f_1^\pm
          + f_0^\pm]$. Denote by $\tilde{\mathbb F}^\pm$ the
          operator generated by its closure in $D[f_1^\pm
          + f_0^\pm]$. For $v \in C_0^{\infty}(A,\infty)$ set
          $$
          w(x, y) : = (F_0 v)(x+iy), \quad x+iy \in \ce.
          $$
          Then we have
          \bel{jul8}
          f_0^\pm[v] = \int_{\Omega_\pm} e^{-bx^2} |w|^2 dxdy, \quad
          f_1^\pm[v] = 2\int_{\Omega_\pm}e^{-bx^2} |\nabla w|^2 dxdy.
          \ee
          Since the $\Omega_\pm$ is a bounded domain with a
          Lipschitz boundary, the Sobolev space ${\rm
          H}^1(\Omega_\pm)$ is compactly embedded in
          $L^2(\Omega_\pm)$. Hence, \eqref{jul8} implies that $\tilde{\mathbb
          F}^\pm$ is compact. \\
          Let us now check that $\|\tilde{\mathbb
          F}^\pm\| < 1$. Evidently, $\|\tilde{\mathbb
          F}^\pm\| \leq 1$. Assume $\|\tilde{\mathbb
          F}^\pm\| = 1$. Since $\tilde{\mathbb
          F}^\pm$ is compact, this means that there exists $0 \neq
          v^\pm \in D[f_1^\pm
          + f_0^\pm]$ such that $f_1^\pm[v] = 0$.
Let $\left\{v_n^\pm\right\}_{n \in {\mathbb N}}$ be sequence of
functions $v_n^{\pm} \in C_0^{\infty}(A,\infty) \subset
C_0^{\infty}(\re)$ converging to $v^\pm$ in $D[f_1^\pm
          + f_0^\pm]$.
           Set
$w_n^\pm(z) =
          (F_0 v_n^\pm)(z)$. Evidently, for any $n \in \N$ we have $w_n^\pm \in {\mathcal B}(\Omega_\pm)$
          (see \eqref{jul20}).  Since ${\mathcal B}(\Omega_\pm)$ is complete,
          there exists $w^\pm \in {\mathcal B}(\Omega_\pm)$
          such that $\lim_{n \to \infty} \|w_n^\pm - w^\pm\|_{{\mathcal
          B}(\Omega_\pm)} = 0$. Since $(F_1 v_n^\pm)(z) = \frac{\partial
          w_n^\pm}{\partial z}$,
          it is not difficult to check that $f_1^\pm[v^\pm] = 0$ implies that $w^\pm$ is
          constant in $\Omega_\pm$ (see e.g. \cite[Theorem 2, Exercise 1]{ha}), and hence $w^\pm$ admits a unique
          analytic extension as a constant to $\ce$. Then the
          distributional Paley-Wiener theorem (see e.g. \cite[Theorem 1.7.7]{hor})
          combined with \cite[Theorem V.11]{rs} implies that $v^\pm$ is proportional to
          the Dirac $\delta$-function supported at $k=0$. Since
          ${\rm supp}\,v^\pm \subset [A, \infty)$ and $A>0$ we conclude that $v^\pm = 0$ as an element
          of  ${\mathcal D}'(\re)$,
          $f_1^\pm[v^\pm]
          + f_0^\pm[v^\pm] = 0$ which contradicts with the hypothesis that $v^\pm \neq 0$ as an element of $D[f_1^\pm
          + f_0^\pm]$. Therefore, $\|\tilde{\mathbb
          F}^\pm\| < 1$, and the quadratic form
          $f_0^\pm$ is bounded, and hence closable in
          $D[f_1^\pm]$. Finally, the operator ${\mathbb
          F}^\pm$ generated by its closure is unitarily equivalent
          to $(I - \tilde{\mathbb
          F}^\pm)^{-1} \tilde{\mathbb
          F}^\pm$ and therefore is compact in $D[f_1^\pm]$.
    \end{proof}
    \begin{pr} \label{sanpr3} For every $r>0$, $A>0$, $\delta \in (0,
    1/2)$, and $\varepsilon \in (0,1)$,  we have
    \bel{jun5}
    n_*(r;  Q_j^-(\lambda; A)) \geq
n_*(r(1+\varepsilon) \sqrt{W_+ - W_-}; \Gamma^-_\delta(\sqrt{b
|\ln{\lambda}|})) + O(1),
    \ee
    \bel{jun6}
    n_*(r;  Q_j^+(\lambda; A)) \leq
n_*(r(1-\varepsilon)  \sqrt{W_+ - W(-\delta)} e^{-b\delta^2/2};
\Gamma^+_\delta({  \sqrt{ b |\ln{\lambda}|}}) + O(1),
    \ee
    as $\lambda \downarrow 0$.
    \end{pr}
\begin{proof}
     Let $\lambda > 0$, $A \in [-\infty, \infty)$.
    Define the operators $M_{j,1}^\pm(\lambda; A) : L^2(\Omega_\pm) \to L^2(\Omega_\pm)$
    as the operators with integral kernels
    $$
    \frac{p_j}{2\pi} e^{-b(x^2 + x'^2)/2} \int_A^{\infty}
    ({\mathcal E}_j^+ - E_j(k;W_0^\pm) + \lambda)^{-1} k^{2(j-1)} e^{-b^{-1}k^2} e^{k(x+x' +
    i(y-y'))}dk
    $$
with $(x,y), \, (x',y') \in \Omega_\pm$. Evidently,
$Q_j^\pm(\lambda; A) Q_j^\pm(\lambda; A)^* = M_{j,1}^\pm(\lambda;
A)$. Therefore,
    \bel{san4}
    n_+(r; Q_j^\pm(\lambda; A)^* Q_j^\pm(\lambda; A)) =
     n_+(r;  M^\pm_{j,1}(\lambda; A)), \quad r>0.
    \ee
    Further, we concentrate at the proof of \eqref{jun5}. Fix $\varepsilon > 0$.
    Then by \eqref{te10} and \eqref{lau25}, there exists $A^-_0 = A^-_0(\varepsilon)$ such that $k \geq A^-_0$ implies
    \bel{jun8}
    {\mathcal E}_j^+ - E_j(k; W_0^-) \leq (1+\varepsilon) \frac{W_+ - W_-}{2} p_j k^{2j-3}e^{-b^{-1}k^2}.
    \ee
    For $p>0$ and $A>0$ define $M_{j,2}^-(\lambda,A,p) : L^2(\Omega_-) \to L^2(\Omega_-)$
    as the operator with integral kernel
    \bel{san5}
     \frac{p_j}{2\pi} e^{-b(x^2 + x'^2)/2} \int_A^{\infty} (p + \lambda k^{3-2j}  e^{b^{-1}k^2})^{-1}
     k e^{k(x+x' + i(y-y'))}dk
    \ee
 with $(x,y), \, (x',y') \in \Omega_-$.
    Then \eqref{jun8} implies that for $A_1 = \max\{A, A_0^-\}$ we have
    \bel{k1}
    n_+(r; M_{j,1}^-(\lambda; A)) \geq n_+\left(r; M_{j,2}^-(\lambda, A_1, p_j(1+\varepsilon)
    (W_+ - W_-)/2)\right).
    \ee
    Fix $\delta \in (0,1/2)$. Set $\Lambda : = |\ln{\lambda}|^{1/2}$, and assume that $\lambda > 0$ is small that
    $A_1 <  \delta \sqrt{b}  \Lambda$. Then, by the mini-max
    principle,
    \bel{k2}
    n_+(r; M_{j,2}^-(\lambda, A_1, p)) \geq n_+(r; M_{j,2}^-(\lambda,  \delta \sqrt{b}  \Lambda, p)),
    \quad p > 0, \quad r > 0.
    \ee
    In the integral defining the kernel of the operator
    $M_{j,2}^-(\lambda,  \delta \sqrt{b}  \Lambda, p)$ (see \eqref{san5}), change the variable
    $k = \sqrt{b} \Lambda (1 + u)^{1/2}$ with $u \in (-1+\delta^2, \infty)$. Then we see that the
    integral kernel of $M^-_{j,2}(\lambda,  \delta \sqrt{b}  \Lambda, p)$ is equal to
    $$
     \frac{p_j b  \Lambda^2}{4\pi}  e^{-b(x^2 + x'^2)/2}\int_{-1+\delta^2}^{\infty} {  (p +  (\sqrt{b}
     \Lambda (1 + u)^{1/2})^{3-2j} e^{\Lambda^2 u})^{-1} }e^{(x+x' + i(y-y'))\sqrt{b}
     \Lambda (1 + u)^{1/2}}du.
    $$
    Define $M_{j,3}^-(\lambda,\delta,p) : L^2(\Omega_-) \to
L^2(\Omega_-)$ as the operator with integral kernel
    $$
     \frac{p_j b  \Lambda^2}{4\pi} e^{-b(x^2 + x'^2)/2}  \int_{-1+\delta^2}^{-1 + (1-\delta)^2}
   {   (p +  (\sqrt{b} \Lambda (1 + u)^{1/2})^{3-2j} e^{\Lambda^2 u})^{-1}}
     e^{(x+x' + i(y-y'))\sqrt{b} \Lambda (1 + u)^{1/2}}du
    $$
    with $(x,y), \, (x',y') \in \Omega_-$. Evidently, the mini-max
    principle  implies
    \bel{k3}
    n_+(r; M_{j,2}^-(\lambda,  \delta \sqrt{b} \Lambda, p)) \geq n_+(r; M_{j,3}^-(\lambda,\delta, p)),
    \quad p > 0, \quad r > 0, \quad \delta \in (0,1/2).
    \ee
    Further, define $M_{j,4}^-(\lambda,\delta,p) : L^2(\Omega_-) \to L^2(\Omega_-)$ as the operator
    with integral kernel
    \bel{san6}
     \frac{p_j b  \Lambda^2}{4\pi p} e^{-b(x^2 + x'^2)/2} \int_{-1+\delta^2}^{-1 + (1-\delta)^2}
      e^{(x+x' + i(y-y'))\sqrt{b} \Lambda (1 + u)^{1/2}}du
     \ee
    with $(x,y), \, (x',y') \in \Omega_-$. By the dominated convergence theorem,
    $$
    \lim_{\lambda \downarrow 0} \|M_{j,3}^-(\lambda,\delta,p) - M_{j,4}^-(\lambda,\delta,p)\|_2^2 = 0
    $$
    where $\|\cdot \|_2$ denotes the Hilbert-Schmidt norm. Fix $\varepsilon > 0$. Applying the Weyl inequalities
    and the elementary Chebyshev-type estimate
    $$
    n_*(s; M_{j,3}^-(\lambda,\delta,p) - M_{j,4}^-(\lambda,\delta,p))
    \leq s^{-2} \|M_{j,3}^-(\lambda,\delta,p) - M_{j,4}^-(\lambda,\delta,p)\|_2^2, \quad s>0,
    $$
    we get
    \bel{k4}
    n_+(r; M_{j,3}^-(\lambda,\delta,p)) \geq n_+(r(1+\varepsilon); M_{j,4}^-(\lambda,\delta,p)) + O(1),
    \quad \lambda \downarrow 0.
    \ee
    In the integral defining the kernel of the operator
$M_{j,4}^-(\lambda,  \delta, 1)$ (see \eqref{san6}),
    change the variable $(1 + u)^{1/2} = k$ with $k \in (\delta, 1-\delta)$.
    Then we see that the integral kernel of $M_{j,4}^-(\lambda,  \delta, 1)$ equals
    $$
     \frac{p_j b  \Lambda^2}{2\pi p} e^{-b(x^2+x'^2)/2} \int_{\delta}^{1-\delta}
     e^{(x+x' + i(y-y'))\sqrt{b} \Lambda k} k dk, \quad (x,y), \, (x',y') \in \Omega_-.
         $$
    Therefore
    \bel{k5}
M_{j,4}^-(\lambda,\delta,p) = \frac{p_j}{2 p}
\Gamma_\delta^-(\sqrt{b |\ln{\lambda}|}) \Gamma_\delta^-(\sqrt{b
|\ln{\lambda}|})^*.
 \ee
 Combining now \eqref{san4}, \eqref{k1}, \eqref{k2}, \eqref{k3},
 \eqref{k4}, and \eqref{k5}, we obtain \eqref{jun5}. \\
 Let us now prove \eqref{jun6}. The proof is quite similar to that
 of  {\eqref{jun5}}, so that we omit certain details. Set $\nu_1 = 0$ and $\nu_j = 1$ if $j \in \N$, $j
 \geq 2$. Pick $\varepsilon \in (0,1)$. Then there exists $A^+_0 = A^+_0(\varepsilon)$ such that $k \geq A^+_0$ implies
    \bel{jun9}
    {\mathcal E}_j^+ - E_j(k; W_0^+) \geq
    (1-\varepsilon) \frac{W_+ - W(-\delta)}{2} p_j (k+ \nu_j)^{2j-3}e^{-(b^{-1/2}k + b^{1/2}\delta)^2}.
    \ee
    For $p>0$ and $A>0$ define $M_{j,2}^+(\lambda,A,p) : L^2(\Omega_+) \to L^2(\Omega_+)$
    as the operator with integral kernel
    \bel{jun10}
     \frac{p_j}{2\pi} e^{-b(x^2 + x'^2)/2} \int_A^{\infty} (p + \lambda (k+\nu_j)^{3-2j}  e^{b^{-1}k^2+2\delta k})^{-1}
     k e^{k(x+x' + i(y-y') + 2\delta)}{dk},
    \ee
 with $(x,y), \, (x',y') \in \Omega_+$.
    Therefore, similarly to \eqref{k1}, we have
    \bel{jun11}
    n_+(r; M_{j,1}^+(\lambda; A)) {\leq} n_+\left(r; M_{j,2}^+(\lambda,
    A_1,
    (1-\varepsilon) p_j {e^{-b\delta^2}}
    (W_+ - W(-\delta))/2)\right)
    \ee
    for $A_1 = \max\{A, A_0^+\}$. Moreover, it is easy to check
    that
     \bel{jun12}
 n_+(r; M_{j,2}^+(\lambda, A, p)) = n_+(r; M_{j,2}^+(\lambda, 0,
 p)) + O(1), \quad \lambda \downarrow 0,
    \ee
    for any $A \geq 0$, $p>0$.
    In the integral defining the kernel of the operator
    $M_{j,2}^+(\lambda,  0, p)$ (see \eqref{jun10}), change the variable
    $k = \sqrt{b} \Lambda (1 + u)^{1/2}$ with $u \in (-1, \infty)$. Then we see that the
    integral kernel of { $M^+_{j,2}(\lambda, 0, p)$} is equal to
    $$
     \frac{p_j b  \Lambda^2}{4\pi} e^{-b(x^2 + x'^2)/2}
    $$
    $$
      \int_{-1}^{\infty} { (p + (\sqrt{b}
     \Lambda (1 + u)^{1/2} + \nu_j)^{3-2j} e^{\Lambda^2 u + 2 \delta \sqrt{b} \Lambda (1 + u)^{1/2} })^{-1}}
     e^{(x+x' + i(y-y') + 2\delta)\sqrt{b}
     \Lambda (1 + u)^{1/2}}du.
    $$
Define now $M_{j,3}^+(\lambda,\delta,p) : L^2(\Omega_+) \to
L^2(\Omega_+)$,  as the operator
    with integral kernel
    $$
     \frac{p_j b  \Lambda^2}{4\pi} e^{-b(x^2 + x'^2)/2}
    $$
    $$
     \int_{-1}^{-1 + (1+\delta)^2}
     {(p +  (\sqrt{b}
     \Lambda (1 + u)^{1/2} + \nu_j)^{3-2j} e^{\Lambda^2 u + 2 \delta \sqrt{b} \Lambda (1 + u)^{1/2} })^{-1}}
     e^{(x+x' + i(y-y') + 2\delta)\sqrt{b}
     \Lambda (1 + u)^{1/2}}du \,
    $$
    with $(x,y), \, (x',y') \in \Omega_+$. By the dominated convergence theorem,
    $$
    \lim_{\lambda \downarrow 0} \|M_{j,2}^+(\lambda,\delta,p) - M_{j,3}^+(\lambda,\delta,p)\|_2^2 =
    0.
    $$
    Therefore, similarly to \eqref{k4}, {we obtain}
    \bel{jun13}
    n_+(r; M_{j,2}^+(\lambda,\delta,p)) \leq n_+(r(1-\varepsilon); M_{j,3}^+(\lambda,\delta,p)) + O(1),
    \quad \lambda \downarrow 0,
    \ee
    for any $r>0$, $\varepsilon \in (0,1)$, $\delta > 0$, $p>0$. Next, define
    $M_{j,4}^+(\lambda,\delta,p) : L^2(\Omega_+) \to L^2(\Omega_+)$, $\delta>0$, as the operator
    with integral kernel
    $$
     \frac{b p_j \Lambda^2}{4\pi p} e^{-b(x^2 + x'^2)/2}
     \int_{-1}^{-1 + (1+\delta)^2}
e^{(x+x' + i(y-y') + 2\delta)\sqrt{b}
     \Lambda (1 + u)^{1/2}}du, \quad (x,y), \, (x',y') \in \Omega_+.
    $$
     Evidently, the mini-max
    principle implies
    \bel{jun14}
    n_+(r; M_{j,3}^+(\lambda,\delta,p)) \leq n_+(r; M_{j,4}^+(\lambda,\delta,p)),
    \quad r>0.
    \ee
    Finally, by analogy with \eqref{k5}, we get
    \bel{jun15}
M_{j,4}^+(\lambda,\delta,p) = \frac{p_j}{2 p}
\Gamma_\delta^+(\sqrt{b |\ln{\lambda}|}) \Gamma_\delta^+(\sqrt{b
|\ln{\lambda}|})^*.
 \ee
 Putting together \eqref{san4} and \eqref{jun11} -- \eqref{jun15},
 we arrive at \eqref{jun5}.

 \end{proof}

 Now, the combination of \eqref{sof19} and \eqref{jun5} -
 \eqref{jun6} yields \eqref{jun3a} - \eqref{jun3}.

\section{Asymptotic bounds of ${\mathcal N}_j^+(\lambda)$ as $\lambda \downarrow 0$}
\label{section6} \setcounter{equation}{0}
    In what follows we identify when appropriate $\rd$ with $\ce$
    writing $z = x + iy \in \ce$ for $(x,y) \in \rd$. Moreover, we
    denote by $d\mu(z) = dxdy$ the Lebesgue measure on $\rd$.
    Further,  we assume as before that $x^+ =
    0$, that $V$ satisfies \eqref{jun2} with some constants
    $c_0^{\pm} > 0$ and some bounded domains $\Omega_\pm \subset
    \rd$ with Lipschitz boundaries, and that
    \bel{jun20}
    \Omega_- \cap \{z \in \ce \, | \, {\rm Re}\,z > 0\} \neq
    \emptyset.
    \ee
    We will show that {under} these assumptions the functions ${\mathcal
    N}_j^+(\lambda)$ satisfy the asymptotic estimates \eqref{lau2}
    and \eqref{lau3} with some explicit constants ${\mathcal
    C}_\pm > 0$. In order to define these constants we need the
    following notations. Let $\Omega \subset \rd$ be a bounded
    domain. Set
    $$
    K_-(\Omega) : = \left\{(p,q) \in \rd \, | \, p<q, \exists \, x \in \re \quad \mbox{such that} \quad
    (x, p + t(q-p)) \in \Omega,
    \forall t \in (0,1)\right\},
    $$
    $$
    {\bf c}_-(\Omega) : = \sup_{(p,q) \in {K_-(\Omega)}} (q-p).
    $$
    In other words, ${\bf c}_-(\Omega)$ is just the maximal length
    of the vertical segments contained in $\overline{\Omega}$.
    Next, for $s \in [0,\infty)$ put
    $$
    \varkappa(s) : = |\{t> 0 \, | t \ln{t} <s\}|
    $$
    where $|\cdot|$ denotes  the Lebesgue measure of a Borel set in
    $\re$. Let $B_R(\zeta) \subset \rd$ be the open disk of radius $R>0$
    centered at $\zeta \in \ce$. Set
    $$
    K_+(\Omega) : = \left\{(\xi,R) \in \re \times (0,\infty) \, | \, \exists \, \eta \in \re \quad \mbox{such that} \quad
    \Omega \subset B_R(\xi + i\eta)\right\},
    $$
    $$
    {\bf c}_+(\Omega) : = \inf_{(\xi,R) \in K_+(\Omega)} R \varkappa\left(\frac{\xi_+}{eR}\right)
    $$
 where $\xi_+ : = \max\{\xi,0\}$. Evidently,
    \bel{jul50}
 {\bf c}_+(\Omega)
 \geq \frac{1}{2} {\rm diam}\,(\Omega) \geq \frac{1}{2} {\bf
 c}_-(\Omega).
 \ee
 Finally, put
 $$
    \tilde{\Omega}_\pm : = \{z \in \Omega_\pm \, | \, {\rm Re}\,z > 0\}.
    $$
    Note that \eqref{jun20} implies $\tilde{\Omega}_\pm \neq
    \emptyset$. Occasionally, we will also use the notation
    $$
    \tilde{\Omega}_+(\delta) : = \{z \in {\Omega_+} \, | \, {\rm Re}\,z > -2\delta\}
    $$
    for $\delta \geq 0$ so that $\tilde{\Omega}_+(0) =
    \tilde{\Omega}_+$.
    \begin{theorem} \label{junth}
    Suppose that $W$ is a bounded non-decreasing function with
    $W_- < W_+$, and $x^+ = 0$. Assume that {$V$} satisfies
    \eqref{jun2}, and \eqref{jun20} holds true. Then asymptotic
    relation \eqref{lau2} is satisfied with ${\mathcal C}_- : =
     (2\pi)^{-1} \sqrt{b} {\bf c}_-(\tilde{\Omega}_-)$ while asymptotic
    relation \eqref{lau3} holds true with ${\mathcal C}_+ : =
    e \sqrt{b} {\bf c}_+(\tilde{\Omega}_+)$. In particular,
    $$
    \lim_{\lambda \downarrow 0} \frac{\ln{{\mathcal N}_j^+(\lambda)}}{\ln\,|{\ln{\lambda}}|} =
    \frac{1}{2}, \quad j \in \N.
    $$
\end{theorem}
{\em Remark.} Under the hypotheses of Theorem \ref{junth} we have
${\mathcal C}_- < {\mathcal C}_+$ due to \eqref{jul50},
$\tilde{\Omega}_- \subset \tilde{\Omega}_+$, and $1/\pi < e$. \\

 The proof of \eqref{lau2} is contained in Subsection
\ref{ss61}, and the proof of \eqref{lau3} can be found in
Subsection \ref{ss62}.

\subsection{Lower bound of ${\mathcal N}_j^+(\lambda)$}
\label{ss61}
    In this subsection we prove \eqref{lau2}. Taking into account
    {Theorem \ref{lauth1}, \eqref{lau60} and Theorem \ref{lauth2}}, we find that it
    suffices to show that for any $r>0$ independent of $\lambda > 0$, we have
    \bel{jul51}
    \lim_{\delta \downarrow 0} \liminf_{\lambda \downarrow 0}
    |\ln{\lambda}|^{-1/2} n_+(r;
    \Gamma_{\delta}^-(\sqrt{b|\ln{\lambda}|})^* \Gamma_{\delta}^-(\sqrt{b|\ln{\lambda}|})) \geq {\mathcal C}_-.
    \ee
    Let $\Omega \subset \rd$ be a bounded domain, and ${\mathcal
    I} \subset (0,\infty)$ be a bounded open non-empty interval.
    For $m>0$ and $\delta \geq 0$ define the operator ${\mathcal G}_{m, \delta}(\Omega,
    {\mathcal I}) : L^2({\mathcal I}) \to L^2({\mathcal I})$ as
    the operator with integral kernel
    \bel{jun22}
    \pi^{-1} m^2 \sqrt{k k'} \int_\Omega e^{m((z+\delta)k + (\bar{z} + \delta) k')}
    d\mu(z), \quad k, k' \in {\mathcal I}.
    \ee
    Set
    \bel{jul52}
    \epsilon_- : = \inf_{x \in \omega_-} e^{-bx^2}, \quad
     \epsilon_+ : = \sup_{x \in \omega_+} e^{-bx^2},
     \ee
     the numbers $\omega_{\pm}$ being defined in \eqref{aug1}.
     Then we have
     \bel{jul53}
     \Gamma_{\delta}^-(m)^*
     \Gamma_{\delta}^-(m) \geq \epsilon_- {\mathcal G}_{m,0}(\Omega_-,
    I_-(\delta)), \quad m>0.
    \ee
    Further, let ${\mathcal R} \subset \tilde{\Omega}_- \subset
    \Omega_-$be an open non-empty rectangle whose sides are
    parallel to the coordinate axes. Since a translation $z
    \mapsto z + i\eta$, $\eta \in \re$, in the integral in
    \eqref{jun22} generates a unitary transformation of the operator ${\mathcal G}_{m,0}(\Omega_-,
    I_-(\delta))$ into an operator unitarily equivalent to it, we
    assume without any loss of generality that ${\mathcal R} =
    (\alpha, \beta) \times (-L,L)$ with $0<\alpha<\beta<\infty$
    and $L \in (0,\infty)$. Evidently,
    \bel{jul54}
    {\mathcal G}_{m,0}(\Omega_-,
    I_-(\delta)) \geq {\mathcal G}_{m,0}({\mathcal R},
    I_-(\delta)), \quad m>0.
    \ee
    For $\eta \in \re$ and $\delta \in (0,1/2)$ define the operator
    $G^-_{\eta, \delta}(m) : L^2(I_-(\delta)) \to L^2(I_-(\delta))$ as the integral operator
    with kernel
    $$
    e^{\eta m(k + k')} \frac{\sin{(m(k-k'))}}{\pi (k-k')} \frac{2\sqrt{k k'}}{k+k'}, \quad k, k' \in I_-(\delta).
    $$
  Then
    \bel{jul55}
    {\mathcal G}_{m,0}({\mathcal R},
    I_-(\delta)) =
     G^-_{\beta, \delta}(mL) - G^-_{\alpha,
    \delta}(mL).
    \ee
Define the
    operator $g_{\mathcal I}(m): L^2(\mathcal I) \to L^2(\mathcal
    I)$, $m>0$,
    as the operator with integral kernel
    $$
\frac{\sin{(m(k-k'))}}{\pi(k-k')} \frac{2\sqrt{k k'}}{k+k'}, \quad
k,k' \in {\mathcal I}.
    $$
   Evidently, $g_{\mathcal I}(m) = g_{\mathcal I}(m)^* \geq 0$ is a trace-class
   operator.
   Simple variational arguments yield
    $$
   n_+(r; {G^-_{\beta, \delta}}({ m}) - {G^-_{\alpha, \delta}}(m)) \geq
   n_+(r (1-e^{2(\alpha - \beta)\delta m})^{-1}; {G^-_{\beta, \delta}}(m)) \geq
$$
\bel{repon4}
     n_+(r e^{-2\beta \delta m }(1-e^{2(\alpha - \beta)\delta m })^{-1}; {g_{I_-(\delta)}}(m)),
    \quad r>0, \quad \delta \in (0,1/2).
    \ee
Combining \eqref{jul53} -- \eqref{repon4},
    we find that under the hypotheses of Theorem \ref{junth} for each $\delta \in (0,1/2)$ we have
    \bel{san7}
    n_+(r;
    \Gamma_{\delta}^-(m)^* \Gamma_{\delta}^-(m)) \geq
    n_+(re^{-2\beta \delta m}(\epsilon_-(1-e^{2(\alpha - \beta)\delta m} ))^{-1};
    g_{I_-(\delta)}(mL)).
    \ee
In order to complete the proof of \eqref{jul51}, we need the
following
\begin{pr} \label{sanpr2} For all $l \in {\mathbb N}$ we have
    \bel{renov1}
    \lim_{m \to \infty} m^{-1} {\rm Tr}\,g_{\mathcal I}(m)^l =
    \frac{|\mathcal I|}{\pi}.
    \ee
    \end{pr}
    \begin{proof}
    Let $l=1$. Then,  ${\rm Tr}\,g_{\mathcal I}(m) =
    \frac{m|\mathcal I|}{\pi}$. Let now $l \geq 2$. Set
    $$
\phi_m(k) : = \frac{\sin{m k}}{\pi k }  \quad {k} \in {\mathcal I}.
    $$
    Denote by $\chi_{\mathcal I}$ the characteristic function of the
    interval ${\mathcal I}$. Then we have
    $$
    {\rm Tr}\,g_{\mathcal I}(m)^l =
    $$
    $$
    \int_\re \ldots \int_\re \phi_m(k_1 - k_2) \phi_m(k_2 -
    k_3) \ldots \phi_m(k_{l-1} - k_l) \phi_m(k_l -
    k_1) \times
    $$
    $$
    \frac{2^l \, k_1 \ldots k_l}{(k_1 + k_2)(k_2 + k_3) \ldots (k_{l-1} + k_l)(k_l +
    k_1)} \chi_{\mathcal I}(k_1) \ldots \chi_{\mathcal I}(k_l) dk_1 \ldots dk_l.
    $$
Changing the variables $k_1 = t_1$, ${k_j} = t_1 + m^{-1} t_j$,
$j=2,\ldots,l$, we get
$$
    {\rm Tr}\,g_{\mathcal I}(m)^l =
    $$
    $$
    m \int_\re \ldots \int_\re \phi_1( - t_2) \phi_1(t_2 -
    t_3) \ldots \phi_1(t_{l-1} - t_l) \phi_1(t_l) \, \times
    $$
    $$
    \frac{2^l \, t_1 (t_1 + m^{-1} t_2) \ldots (t_1 + m^{-1} t_l)}
    {(2t_1 + m^{-1}t_2)(2t_1 + m^{-1}(t_2 + t_3)) \ldots (2t_1 + m^{-1}(t_{l-1} + t_l))(2t_1 +
    m^{-1}t_l)} \, \times
    $$
    $$
    \chi_{\mathcal I}(t_1) \chi_{\mathcal I}(t_1 + m^{-1}t_2) \ldots \chi_{\mathcal I}(t_1 + m^{-1}t_l) dt_1 \ldots dt_l.
    $$
    Applying the dominated convergence theorem, we get
    \bel{renov2}
\lim_{m \to \infty} m^{-1} {\rm Tr}\,g_{\mathcal I}(m)^l =
|{\mathcal I}| \int_\re \ldots \int_\re \phi_1( - t_2) \phi_1(t_2
-
    t_3) \ldots \phi_1(t_{l-1} - t_l) \phi_1(t_l) dt_2 \ldots
    dt_l.
    \ee
    Further, we have
    $$
    \phi_1(t) = \frac{1}{2\pi} \int_\re  e^{it\xi}
    \chi_{(-1,1)}(\xi) d\xi, \quad t \in \re.
    $$
    Therefore,
   \bel{renov3}
\int_\re \ldots \int_\re \phi_1( - t_2) \phi_1(t_2 -
    t_3) \ldots \phi_1(t_{l-1} - t_l) \phi_1(t_l) dt_2 \ldots
    dt_l =
    \frac{1}{2\pi} \int_\re
    \chi_{(-1,1)}(\xi)^l d\xi = \frac{1}{\pi}.
    \ee
Putting together \eqref{renov2} and \eqref{renov3}, we obtain
\eqref{renov1}.
    \end{proof}
    The Kac-Murdock-Szeg\H{o} theorem (see e.g. \cite{kacmsz}, \cite{grsz} or
    \cite{r4}) now implies the following
    \begin{follow} \label{refnov1}
    We have
\bel{renov5}
    \lim_{m \to \infty} m^{-1} n_+(s; g_{\mathcal I}(m))  =
    \left\{
    \begin{array} {l}
    \frac{|{\mathcal I}|}{\pi} \quad {\rm if} \quad s \in (0,1), \\
    0 \quad {\rm if} \quad s > 1.
    \end{array}
    \right.
    \ee
    \end{follow}
    Now we are in position to prove \eqref{jul51}. Fix arbitrary $s \in (0,1)$.
     Assume that $m$ is so large that
    $re^{-2\beta \delta m}(\epsilon_-(1-e^{2(\alpha - \beta)\delta m}))^{-1} <s$.
Then \eqref{san7} implies
    \bel{renov6}
     n_+(r;
    \Gamma_{\delta}^-(m)^* \Gamma_{\delta}^-(m)) \geq
    n_+(s;
    g_{I_-(\delta)}(mL){ )}.
    \ee
Putting together \eqref{renov5}  and \eqref{renov6}, we find that
the
    asymptotic estimate
    $$
\liminf_{\lambda \downarrow 0} |\ln{\lambda}|^{-1/2} n_+(r;
    \Gamma_{\delta}^-(\sqrt{b |\ln{\lambda}|})^* \Gamma_{\delta}^-(\sqrt{b |\ln{\lambda}|})) \geq
    \frac{\sqrt{b}L}{\pi}(1-2\delta)
    $$
    holds for every $\delta \in (0,1/2)$. Letting $\delta \downarrow 0$, and optimizing with respect to $L$ we
    obtain \eqref{jul51}.

\subsection{Upper bound of ${\mathcal N}_j^+(\lambda)$}
\label{ss62}
    In this subsection we prove \eqref{lau3}. By analogy with \eqref{jul51},  it
    suffices to show that for any $r>0$ independent of $\lambda > 0$, we have
    \bel{jul60}
    \lim_{\delta \downarrow 0} \limsup_{\lambda \downarrow 0}
    |\ln{\lambda}|^{-1/2} n_+(r;
    \Gamma_{\delta}^+(\sqrt{b|\ln{\lambda}|})^* \Gamma_{\delta}^+(\sqrt{b|\ln{\lambda}|})) \leq {\mathcal C}_+.
    \ee
    Evidently,
    \bel{jul61}
     \Gamma_{\delta}^+(m)^* \Gamma_{\delta}^+(m) \leq \epsilon_+
     {\mathcal G}_{m,\delta}(\Omega_+; I_+(\delta)), \quad m>0,
     \ee
     the integral kernel of the operator ${\mathcal G}_{m,\delta}(\Omega_+; {\mathcal
     I})$ being defined in \eqref{jun22}, and the number
     $\epsilon_+$ being defined in \eqref{jul52}. Since we have
     ${\mathcal G}_{m,\delta}(\Omega_+\setminus\tilde{\Omega}_+(\delta);
     I_+(\delta))\geq 0$ and
     $$
     \lim_{m \to \infty} {\rm Tr}\,{\mathcal G}_{m,\delta}(\Omega_+\setminus\tilde{\Omega}_+(\delta);
     I_+(\delta)) =
     \pi^{-1} \lim_{m \to \infty}  m^2 \int_0^{1+\delta} {\int_{\Omega_+\setminus\tilde{\Omega}_+(\delta)}}
e^{2m({\rm Re}\,z+\delta)k}
    d\mu(z) k dk = 0,
    $$
    we easily find that the Weyl inequalities entail
    \bel{jul62}
    n_+(r;{\mathcal G}_{m,\delta}(\Omega_+;
     I_+(\delta))) \leq n_+(r(1-\varepsilon);{\mathcal G}_{m,\delta}(\tilde{\Omega}_+(\delta);
     I_+(\delta)) {)}+ O(1), \quad m \to \infty,
     \ee
     for each $r>0$ and $\varepsilon \in (0,1)$. Further, pick an
     open
     disk $B_R(\zeta) \subset \rd$ such that
     $\tilde{\Omega}_+(\delta)\subset B_R(\zeta)$. Evidently,
    \bel{jul63}
    n_+(r;{\mathcal G}_{m,\delta}(\tilde{\Omega}_+(\delta);
     I_+(\delta)) { )} \leq n_+(r;{\mathcal G}_{m,\delta}(B_R(\zeta);
     I_+(\delta))), \quad r>0.
     \ee
     Next, put $I_* = I_*(\delta) : = (0, (1+\delta)^{-1})$, and define
     $G_\delta^+(m) : L^2(I_*) \to L^2(I_*)$ as the operator with
     integral kernel
     $$
\pi^{-1} m^2 e^{2m(\xi + \delta)_+} \int_{B_R(0)} e^{m(zk +
\bar{z} k')}
    d\mu(z), \quad k, k' \in {{ I}_*(\delta)}.
    $$
    Changing the variable $z \mapsto z + \zeta$ in the integral
    defining the kernel of ${\mathcal G}_{m,\delta}(B_R(\zeta);
     I_+(\delta))$ (see \eqref{jun22}), and after that changing
     the variable $k \mapsto (1+\delta)^2 k$ in $I_*(\delta)$, we
     find that the mini-max principle implies
     \bel{jul64}
     n_+(r;{\mathcal G}_{m,\delta}(B_R(\zeta);
     I_+(\delta))) \leq n_+(r; G_\delta^+((1+\delta)^2m){ )}, \quad r>0,
     \ee
     {with $\xi=  {\rm Re} \, \zeta$.}
     Further, a simple explicit calculation yields
     $$
\int_{B_R(0)} e^{m(zk + \bar{z} k')} d\mu(z) = \pi R^2
\sum_{q=0}^{\infty} \frac{(m^2 R^2 k k')^q}{(q!)^2 (q+1)}.
$$
Therefore, the quadratic form of the operator $G_\delta^+(m)$ can
be written as
 \bel{jun38}
\langle G_\delta^+(m) u, u\rangle_{L^2(I_*)} = e^{2m(\xi +
\delta)_+} \sum_{q=0}^{\infty} \frac{(m R)^{2q+2}}{(q!)^2 (q+1)}
|\tilde{u}_q|^2
 \ee
where
$$
\tilde{u}_q = \int_{I_*(\delta)} k^q u(k) dk, \quad u \in
L^2(I_*(\delta)), \quad q \in {\mathbb Z}_+.
$$
Let $\{p_q(k)\}_{q \in {\mathbb Z}_+}$ be the system of polynomials orthonormal in $L^2(I_*(\delta))$, obtained by the Gram-Schmidt procedure from
$\{k^q\}_{q \in {\mathbb Z}_+}$, $k \in I_*(\delta)$. Then,
$$
k^q = \sum_{l=0}^q \theta_{q,l} p_l(k), \quad k \in I_*(\delta), \quad q \in {\mathbb Z}_+,
$$
with appropriate $\theta_{q,l}$; in what follows we set $\theta_{q,l} = 0$ for $l>q$. Put
$$
u_q = \int_{I_*(\delta)} p_q(k) u(k) dk, \quad u \in
L^2(I_*(\delta)), \quad q \in {\mathbb Z}_+.
$$
Then we have
    \bel{jun38a}
    \tilde{u}_q =  \sum_{l=0}^{\infty} \theta_{q,l} u_l, \quad q \in {\mathbb Z}_+,
    \ee
    and
     \bel{jun39}
     \|u\|^2_{L^2(I_*(\delta))} = \sum_{q=0}^{\infty} |u_q|^2.
     \ee Further, it is easy to check that
     $$
     \sum_{q=0}^{\infty} \sum_{l=0}^{\infty} \theta_{q,l}^2 = \sum_{l=0}^{\infty}\int_{I_*(\delta)} k^{2l} dk =  \sum_{l=0}^{\infty} \frac{(1+\delta)^{-2l-1}}{2l+1} < \infty.
     $$
     Therefore, the operator $\Theta: l^2({\mathbb Z}_+) \to l^2({\mathbb Z}_+)$ defined by
     $$
     (\Theta {\bf u})_q = \sum_{l=0}^{\infty} \theta_{q,l} u_l, \quad q \in {\mathbb Z}_+, \quad {\bf u} = \{u_l\}_{l \in {\mathbb Z}_+},
     $$
     is a Hilbert-Schmidt, and hence bounded operator. Let $\gamma(m): l^2({\mathbb Z}_+) \to l^2({\mathbb Z}_+)$ be the diagonal operator with diagonal entries
     \bel{jun40}
     e^{2m(\xi + \delta)_+} \frac{(mR)^{2q+2}}{(q!)^2 (q+1)}, \quad q \in {\mathbb Z}_+.
     \ee
     Now \eqref{jun38} -- \eqref{jun40} imply
     \bel{jul65}
     n_+(s; G_\delta^+(m)) = n_+(s; \Theta^* \gamma(m) \Theta), \quad s>0.
     \ee
      Evidently,
       \bel{jul66}
      n_+(s; \Theta^* \gamma(m) \Theta) \leq n_+(s; \|\Theta\|^2 \gamma(m)), \quad s> 0.
      \ee
      On the other hand, for any $s>0$ we have
      \bel{jul67}
      n_+(s; \gamma(m)) = \# \left\{q \in {\mathbb Z}_+ \, \left| \, \frac{e^{m(\xi + \delta)_+} (mR)^{q+1}}{q! \sqrt{q+1}} > \sqrt{s}\right.\right\}, \quad s>0.
      \ee
      Applying the Stirling formula
      $$
      q! = (2\pi)^{1/2} (q+1)^{q+1} (q+1)^{-1/2} e^{-q-1} (1 + o(1)), \quad q \to \infty,
      $$
      we find that for each $\varepsilon \in (0,1)$ there exists $q_0 \in {\mathbb Z}_+$ such that
      $$
      \# \left\{q \in {\mathbb Z}_+ \, \left| \, \frac{e^{m(\xi + \delta)_+} (mR)^{q+1}}{q! \sqrt{q+1}} > \sqrt{s}\right.\right\} \leq
      $$
      \bel{jul68}
      \# \left\{q \in {\mathbb Z}_+ \, \left| \, \frac{(\xi + \delta)_+} {e R}  >  \frac{q+1}{eRm} \ln{\left(\frac{q+1}{eRm}\right)} + \frac{\ln{\left(\sqrt{2\pi s} (1-\varepsilon)\right)}}{eRm}\right.\right\} + q_0.
      \ee
      Passing from Darboux sums to Riemann integrals, we find that for each constant $c \in \re$ we have
      $$
      \lim_{m \to \infty} m^{-1}  \# \left\{q \in {\mathbb Z}_+ \, \left| \, \frac{(\xi + \delta)_+} {e R}  >  \frac{q+1}{eRm} \ln{\left(\frac{q+1}{eRm}\right)} + \frac{c}{m}\right.\right\} =
      $$
      \bel{jul69}
      eR \varkappa\left(\frac{(\xi + \delta)_+} {e R}\right).
      \ee
      Putting together \eqref{jul61} -- \eqref{jul64} and \eqref{jul65} -- \eqref{jul69}, we get
      $$
      \limsup_{\lambda \downarrow 0}
    |\ln{\lambda}|^{-1/2} n_+(r;
    \Gamma_{\delta}^+(\sqrt{b|\ln{\lambda}|})^* \Gamma_{\delta}^+(\sqrt{b|\ln{\lambda}|})) \leq
    (1+\delta)^2 \sqrt{b} eR \varkappa\left(\frac{(\xi + \delta)_+} {e R}\right)
    $$
    for any $\delta > 0$. Letting $\delta \downarrow 0$ and optimizing with respect to $\xi$ and  $R$, we obtain \eqref{jul60}. \\

{\large \bf Acknowledgements}. The authors were partially
supported by the Chilean Science Foundation {\em Fondecyt} under
Grant 1090467, and by the Bernoulli Center, EPFL, Lausanne, where
a part of this work was done within the framework of the Program
``{\em Spectral and Dynamical Properties of Quantum Hamiltonians},
January - June 2010.\\ {V. Bruneau was also partially supported by
the {\em Agence Nationale de la Recherche} under Grant NONAa
(ANR-08-BLAN-0228).}\\ P. Miranda and G. Raikov were also
partially supported by {\em N\'ucleo Cient\'ifico ICM} P07-027-F
``{\em Mathematical Theory of Quantum and Classical Magnetic
Systems"}. \\

{\sc Vincent Bruneau }\\
Universit\'e Bordeaux I, Institut de Math\'ematiques de
Bordeaux,\\
UMR CNRS 5251,  351, Cours de la Lib\'eration, 33405 Talence,
France\\
E-Mail: vbruneau@math.u-bordeaux1.fr\\

{\sc Pablo Miranda}\\
Departamento de Matem\'aticas,
Facultad de Ciencias,\\
Universidad de Chile, Las Palmeras 3425, Santiago de Chile\\
E-Mail: pmirandar@ug.uchile.cl\\

{\sc Georgi Raikov}\\
 Departamento de Matem\'aticas, Facultad de
Matem\'aticas,\\ Pontificia Universidad Cat\'olica de Chile,
Vicu\~na Mackenna 4860, Santiago de Chile\\
E-Mail: graikov@mat.puc.cl

\end{document}